\documentclass[aps,reprint,prx,amsmath,amssymb,floatfix]{revtex4-2}

\usepackage[utf8]{inputenc}
\usepackage[T1]{fontenc}
\usepackage{verbatim}
\usepackage[english]{babel}
\usepackage[colorlinks=true, allcolors=blue, linktoc=page]{hyperref}
\usepackage{titlesec}
\usepackage{amsmath,amsthm,amssymb,amscd,amsfonts,mathtools,mathrsfs,braket,bbm,enumerate,graphicx,xcolor,bm}
\numberwithin{equation}{section}
\usepackage{tocloft}
\usepackage{float}
\usepackage{tikz}
\usepackage{algorithm,algpseudocode}

\allowdisplaybreaks
\makeatletter
\renewcommand{\@cfttocfinish}{}
\makeatother

\newtheorem{theorem}{Theorem}
\newtheorem{proposition}[theorem]{Proposition}
\newtheorem{lemma}[theorem]{Lemma}
\newtheorem{corollary}[theorem]{Corollary}

\newtheorem{problem}{Problem}

\theoremstyle{definition}
\newtheorem{definition}{Definition}
\newtheorem{remark}{Remark}

\def\be{\begin{eqnarray}}
\def\ee{\end{eqnarray}}

\DeclareMathOperator*{\argmin}{argmin}
\DeclareMathOperator*{\argmax}{argmax}
\DeclareMathOperator{\tr}{Tr}
\DeclareMathOperator{\poly}{poly}
\DeclareMathOperator{\re}{Re}
\DeclareMathOperator{\im}{Im}
\DeclareMathOperator{\supp}{supp}

\definecolor{Pr}{rgb}{0.4,0.3,0.9}

\begin{document}

\title{The Quantum Esscher Transform}

\author{Yixian Qiu}
\email{yixian_qiu@u.nus.edu}
\affiliation{Centre for Quantum Technologies, National University of Singapore, 3 Science Drive 2, Singapore 117543}

\author{Kelvin Koor}
\email{kelvinkoor@u.nus.edu}
\affiliation{Centre for Quantum Technologies, National University of Singapore, 3 Science Drive 2, Singapore 117543}

\author{Patrick Rebentrost}
\email{patrick@comp.nus.edu.sg}
\affiliation{Centre for Quantum Technologies, National University of Singapore, 3 Science Drive 2, Singapore 117543}
\affiliation{Department of Computer Science, National University of Singapore, 13 Computing Drive, Singapore 117417}

\begin{abstract} 
The Esscher Transform is a tool of broad utility in various domains of applied probability. It provides the solution to a constrained minimum relative entropy optimization problem. In this work, we study the generalization of the Esscher Transform to the quantum setting. We examine a relative entropy minimization problem for a quantum density operator, potentially of wide relevance in quantum information theory. The resulting solution form motivates us to define the \textit{quantum} Esscher Transform, which subsumes the classical Esscher Transform as a special case. Envisioning potential applications of the quantum Esscher Transform, we also discuss its implementation on fault-tolerant quantum computers. Our algorithm is based on the modern techniques of block-encoding and quantum singular value transformation (QSVT). We show that given block-encoded inputs, our algorithm outputs a subnormalized block-encoding of the quantum Esscher Transform within accuracy $\epsilon$ in $\tilde O(\kappa d \log^2 1/\epsilon)$ queries to the inputs, where $\kappa$ is the condition number of the input density operator and $d$ is the number of constraints.
\end{abstract}

\maketitle

\tableofcontents

\section{Introduction}
In probability and statistics, it is often important to find low relative-entropy distributions from a given fixed distribution. In addition, further constraints, the form and interpretation of which depend on the problem at hand, are frequently imposed on the target distribution. 
An interesting example is the following: consider the process of inferring probability distributions from a set of measurement data. The available data play the role of the constraints---they put restrictions on what the true distribution could be---but these data may not suffice to uniquely determine a probability distribution. In this situation, a common approach is to invoke Jaynes' maximum entropy principle (MaxEnt) \cite{jaynes1957information}. In essence, MaxEnt advocates that the selected distribution be the one that simultaneously maximizes entropy and satisfies the given constraints.

However, the situation becomes more nuanced if we already possess some knowledge of the system, say a prior distribution. In such cases, there is a more refined strategy: the minimum relative entropy principle. As expounded in \cite{shore1980axiomatic,olivares2007quantum, zorzi2013minimum}, this principle, regarded as a generalization of MaxEnt, operates by minimizing the distinguishability (characterized by the relative entropy) between the prior distribution and the distribution to be selected, while respecting the imposed constraints. This systematic approach to incorporating new data makes it fundamental in Bayesian statistics. The updating procedure results in the posterior distribution which reflects the most current understanding of the system in light of the observed data.

When the measurement data is presented in the form of expectation values of selected random variables, the solution to the corresponding relative entropy minimization problem takes the form known as an \textit{Esscher Transform}. Named after Swedish mathematician and economist Fredrik Esscher, who introduced the concept in 1932 in his work on risk theory \cite{escher1932probability}, the Esscher Transform, also known as `exponential tilting' in statistics, and its various extensions have since then found many applications beyond minimizing relative entropy. Notable examples include option pricing (in mathematical finance) \cite{gerber1993option}, importance sampling (for rare-event simulation) \cite{siegmund1976importance} and L\'{e}vy processes (in financial economics) \cite{hubalek2006esscher}. More recently, it has also made inroads into machine learning \cite{beirami2023tilted}, in the context of empirical risk minimization.

In this paper, we discuss the extension of the above problem to the quantum setting. We consider the following optimization problem:
\be
\text{minimize}_{\sigma \geq 0} && \;\;S(\sigma\|\rho) \\\nonumber
\text{s.t.} && \;\;\tr (\sigma H_i) = m_i, \quad i \in [d] \\\nonumber
&&  \;\;\tr (\sigma) = 1,
\ee
where $\rho$ is the a priori state and $H_i$, $i \in [d]$ are observables. 
Refer to Definition \ref{Problem:quantum_motivating_problem_statement} for the precise formulation. In the first part of this work, we show the formal solution to this constrained optimization problem. The solution methodology is modelled after its classical predecessor, albeit with added technical intricacies to manage. The form of the corresponding solution then motivates us to define the \textit{quantum} Esscher Transform, see Definition \ref{Definition:quantum_esscher_transform}. The proof of the solution to the optimization problem is found in Theorem \ref{Theorem:partial_solution_to_initial_quantum_problem}.
The quantum Esscher Transform can be viewed as a generalization of the (classical) Esscher Transform, and indeed subsumes the latter as a special case. 
In the second part of this work, with an eye toward potential applications, we discuss the  implementation of the quantum Esscher Transform on fault-tolerant quantum computers. Our algorithm is based on the modern techniques of block-encoding and quantum singular value transformation (QSVT) \cite{gilyen2019quantum,martyn2021grand}. As an input model we consider purifications of the density operator $\rho$ and block-encodings of the operators $H_i$. The main algorithm is Algorithm \ref{Algorithm:QET}, whose complexity is discussed in Theorem \ref{Theorem:Theorem_for_algorithm}.
The quantum Esscher Transform could find applications in quantum analogues of problems in statistics, machine learning, and finance.

\subsection{Preliminaries and notation}
We define the following notations. 
Let $\mathbb{N} = \{1,2,\dots\}$ be the set of positive natural numbers. For $d \in \mathbb{N}$, $[d] = \{1,2,\dots,d\}$. Here $\|\cdot\|$, $\|\cdot\|_1$, $\|\cdot\|_2$ and $\|\cdot\|_T$ refer to the spectral, $l_1$-, $l_2$- and trace norms respectively. 
The symbol $\odot$ denotes component-wise product, e.g. for vectors $(v \odot w)_i = v_iw_i$, for matrices $(A \odot B)_{ij} = A_{ij}B_{ij}$. Throughout this paper, $\log$ will be base $2$. For convenience, when calculus is involved we shall differentiate as if it were base $e$. For a matrix $M$ we write $a \leq M \leq b$ to mean the eigenvalues of $M$ are in $[a,b]$. Thus, $M \geq 0$ means $M$ is positive semidefinite. We denote a Hilbert space by $\mathcal{H}$, $\mathcal{H}_N$ if its dimension $N$ is to be explicitly specified, the set of linear operators on $\mathcal{H}$ by $\mathcal{L}(\mathcal{H})$, and the set of density operators on $\mathcal{H}$ by $\mathcal{D}(\mathcal{H})$. Let $A \in \mathcal{L}(\mathcal{H})$. The kernel of $A$ is $\ker(A):=\{\ket{\psi} \in \mathcal{H}: A\ket{\psi}=0 \}$ and the support of $A$ is $\supp(A):=\ker(A)^\perp$. Note that $\ker(A) \oplus \supp(A) = \mathcal{H}$. $I_n$ denotes the $n$-qubit identity operator, i.e. it is of size $2^n \times 2^n$. We use $\tilde{O}(\cdot)$ to hide polylog factors, i.e., $\tilde{O}(f(n)) := O(f(n)\cdot {\rm polylog}(f(n)))$. We use $A:=B$ to define expression $A$ in terms of $B$.

A probability space is denoted by $(\Omega,\Sigma,P)$, where $\Omega$ is the sample space, $\Sigma$ is the $\sigma$-algebra over $\Omega$, and $P$ is the probability measure on $\Sigma$. While all the discussions in our work are well-defined for general probability spaces, for our purposes we shall restrict our discussion to finite sample spaces, i.e., $|\Omega|<\infty$, and set $\Sigma=2^\Omega$. In this setting, $P$ can be viewed as a $|\Omega|$-dimensional vector residing in the hypercube $[0,1]^{|\Omega|} \subseteq \mathbb{R}^{|\Omega|}$, with components $P(\omega)$, $\omega \in \Omega$ and normalization $\sum_{\omega\in\Omega} P(\omega) = 1$. Note that technically, a probability measure $P$ is a function on the $\sigma$-algebra $\Sigma$, not $\Omega$. Since we are dealing with a finite sample space here, knowing $P(\{\omega\})$ for all $\omega \in \Omega$ gives us full knowledge of $P$, from the additivity property of measures. Thus we can and shall simply view $P$ as a function on $\Omega$ and write $P(\omega)$ in place of $P(\{\omega\})$. Finally, given probability measures $P$ and $Q$, we say $Q$ is absolutely continuous with respect to $P$ (written $Q \ll P$) if $P(\omega)=0 \implies Q(\omega)=0$ for all $\omega$.

\begin{figure}
    \centering
    \includegraphics[width=80mm]{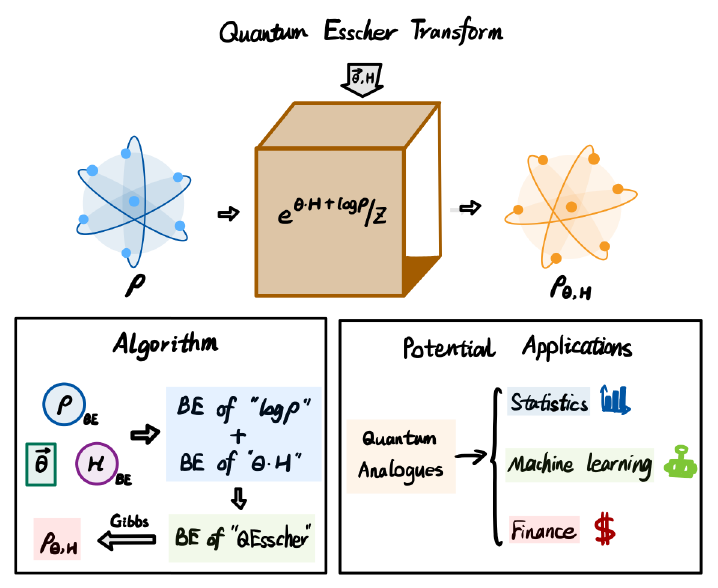}
    \caption{\textbf{Overview of the paper:} Given an input density operator $\rho$, constraint observables $H_i$ and parameters $\theta_i$, $i\in [d]$, the quantum Esscher Transform of $\rho$ is given by $\sigma:=e^{\theta \cdot H + \log \rho}/Z$, where $Z= \tr(e^{\theta \cdot H + \log \rho})$ is the normalization factor. Our algorithm \textsc{QEsscher} is based on block-encodings and QSVT. It takes as inputs block-encodings of \(\rho\), \(H_i\), and the parameter vector \(\vec{\theta}\) and outputs the block-encoding of $\sigma$. The state $\sigma$ itself can also be obtained using Gibbs state preparation techniques in a subsequent step. Finally, we envision potential applications of the quantum Esscher Transform in quantum analogues of problems in statistics, machine learning, and finance.}
    \label{fig:overview}
\end{figure}

\section{Quantum Esscher Transform}

\subsection{Esscher Transform}
The Esscher Transform was first defined by F. Esscher in his work on risk theory \cite{escher1932probability}. 
Let $f: E \longrightarrow \mathbb{R}$ be a probability mass function, where $E \subset \mathbb{R}^d$ and $\theta \in \mathbb{R}^d$. The function
$
f_\theta(x) := \frac{e^{\theta\cdot x}f(x)}{\sum_{x \in E} e^{\theta\cdot x}f(x)}$
is also a probability mass function, and it is called the \textit{Esscher Transform} of $f$ with parameter $\theta$. We can replace probability mass functions with probability density functions (accordingly, $\sum \longrightarrow \int$). 

The Esscher Transform is a map from and onto the space of probability mass/density functions, as $\mathcal{E}(f; \theta) = f_\theta$.
In this work, we never invoke $\mathcal{E}$ and simply call $f_\theta$ the Esscher Transform of $f$, in the same spirit as the Fourier Transform.
In the context of probability theory, let $(\Omega, \Sigma, P)$ be a probability space and  $X : \Omega \longrightarrow \mathbb{R}^d$ a random $d$ dimensional vector. 
This setting motivates the equivalent definition (see Remark \ref{Remark:Equivalence_classical_ET} below) of Esscher Transforms for \textit{measures/distributions}.

\begin{definition}[Esscher Transform for probability distributions]\label{Definition:esscher_transform_for_probability_distributions}
Given a probability distribution $P$ on a finite sample space $\Omega$, a random variable $X: \Omega \longrightarrow \mathbb{R}^d$ and $\theta \in \mathbb{R}^d$. The probability distribution
\begin{align*}
P_{\theta,X}(\omega) := \frac{e^{\theta\cdot X(\omega)}P(\omega)}{\mathbb{E}_P[e^{\theta\cdot X}]}
\end{align*}
is called the Esscher Transform of $P$ with parameter $\theta$, with respect to $X$. For brevity, we say $P_{\theta,X}$ is the $(\theta,X)$-Esscher Transform of $P$.
\end{definition}

This definition is connected to the following problem.  Fix $m \in \mathbb{R}^d$. When and how can we derive from $P$ another probability measure $Q$ such that the expectation of $X$ with respect to $Q$, $\mathbb{E}_Q[X]$ is equal to $m$? Among such probability measures, if they exist, how can we find the one that is closest (in some sense) to $P$? Take as a measure of closeness the relative entropy between $P$ and  $Q$,
$$
D(Q\|P)=\sum_{\omega \in \Omega}Q(\omega)\log \frac{Q(\omega)}{P(\omega)}.
$$
The definition of $D(Q\|P)$ requires that $Q$ be absolutely continuous with respect to $P$, otherwise $D(Q\|P)=\infty$. Without loss of generality, we can assume $P$ is strictly positive on $\Omega$. If this were not so, then let $S \subset \Omega$ denote the subset on which $P=0$. Since $Q$ is absolutely continuous w.r.t. $P$, we have 
$D(Q\|P) = \sum_{\omega \in \Omega \setminus S}Q(\omega)\log \frac{Q(\omega)}{P(\omega)}$,
so we are reduced to an `effective $\Omega$' on which $P$ is strictly positive. The aforementioned question can then be cast as an optimization problem with multiple constraints:
\be
\label{classical_motivating_problem_statement}
\text{minimize}_{Q \in [0,1]^{|\Omega|}} \;\;&& D(Q\|P)\\\nonumber
\text{s.t.} \;\;&& \mathbb{E}_Q[X_i] = m_i, \quad i\in[d]\\\nonumber
&& \sum_{\omega\in \Omega}Q(\omega)=1. 
\ee
Note that there are $d+1$ constraints on $Q$, hence in feasible, non-redundant cases we have $d +1 \leq |\Omega|$, or equivalently $d<|\Omega|$. We have the following solution to the optimization problem. 

\begin{theorem}[]\label{Theorem:solution_to_classical_problem}
Given a random vector $X: \Omega \longrightarrow \mathbb{R}^d$ and $m \in \mathbb{R}^d$ where $\min_{\omega \in \Omega} X_i(\omega) < m_i < \max_{\omega \in \Omega} X_i(\omega)$ for $i \in [d]$ where $d<|\Omega|$. There exists a unique solution $Q^\star$ to problem \ref{classical_motivating_problem_statement}, given by
\begin{align*}
Q^\star=\frac{e^{\lambda^\star \cdot X}P}{\mathbb{E}_P[e^{\lambda^\star \cdot X}]},
\end{align*}
where $\lambda^\star:= \argmin_{\lambda \in \mathbb{R}^d} \mathbb{E}_P[e^{\lambda \cdot (X-m)}]$. Thus $Q^\star$ is the $(\lambda^\star,X)$-Esscher Transform of $P$, see Definition \ref{Definition:esscher_transform_for_probability_distributions}.
\end{theorem}

The proof is elaborated in Appendix \ref{Appendix:Complete_solution_lambda}. Let us comment on a subtlety. Above, we have called $Q^\star$ the Esscher Transform of $P$. Recall that the Esscher Transform as originally defined by Esscher pertains to probability mass/density functions instead of measures. In Remark \ref{Remark:Equivalence_classical_ET}, we show that using the same terminology for probability measures is well-justified, at least for the case when $\Omega$ is discrete.

\subsection{Quantum version}

\textit{Problem statement} $-$
Many concepts in classical probability theory have meaningful generalizations in quantum theory. For example, sample spaces, probability distributions and random variables find their respective counterparts in Hilbert spaces, density operators and observables (the latter also include the former as special instances). The quantum counterpart of the relative entropy is the \textit{quantum relative entropy},
\begin{align*}
S(\sigma\|\rho):= \tr\{\sigma(\log \sigma - \log \rho)\},
\end{align*}
defined for density operators $\sigma,\rho$. As in the classical case, the definition of $S(\sigma\|\rho)$ imposes constraints on $\sigma$ and $\rho$ in order to have $S(\sigma\|\rho) < \infty$. Namely, $\supp(\sigma)\subseteq \supp(\rho)$ (see Chapter 11, \cite{wilde2013quantum}) or equivalently, $\ker(\rho)\subseteq \ker(\sigma)$. Using terminology from measure theory, if this condition is satisfied we say $\sigma$ is absolutely continuous with respect to $\rho$ ($\sigma \ll \rho$). This is analogous to the absolute continuity between probability distributions in classical probability theory.
Now we formally state the quantized version of Problem \ref{classical_motivating_problem_statement}.
\begin{problem}\label{Problem:quantum_motivating_problem_statement}
Let $\mathcal{H}_N$ be an $N$-dimensional Hilbert space and $\rho \in \mathcal{D}(\mathcal{H}_N)$ be a density operator. For $i \in [d]$ where $d<N^2$, let $H_i$ be an observable with $h_{i,\min}$ and $h_{i,\max}$ denoting its smallest and largest eigenvalue respectively. For $m \in \mathbb{R}^d$ with $h_{i,\min} < m_i < h_{i,\max}$, solve
\be
\text{minimize}_{\sigma \geq 0} \;\;&& S(\sigma\|\rho) \\\nonumber
\text{s.t.} \;\;&& \tr (\sigma H_i) = m_i, \quad i \in [d] \\\nonumber
&&  \tr (\sigma) = 1.
\ee
\end{problem}
Here $h_i$ denotes a generic eigenvalue of $H_i$. Note that because $\sigma,H_i$ are Hermitian, $\tr(\sigma H_i)$ is real. As before, we require $h_{i,\min} < m_i < h_{i,\max}$, otherwise the constraints $\tr(\sigma H_i)=m_i$ cannot be satisfied. Finally, we can assume WLOG that $\|H_i\| \leq 1$. This amounts to dividing the constraint $\tr (\sigma H_i) = m_i$ throughout by $\|H_i\|$ if necessary.

\textit{Solution} $-$
Before considering the solution, let us briefly comment on a few possible concerns.
First, $S(\sigma\|\rho)$ requires taking the logarithm of $\rho$, which poses a problem if $\rho$ is not strictly positive definite. This issue is circumvented if, as mentioned above, $\ker(\rho)\subseteq \ker(\sigma)$. The analysis becomes relatively straightforward if we partition the Hilbert space $\mathcal{H}$ into suitable subspaces and examine $\sigma$ over them separately. To this end, we introduce the following notation. Let $\mathcal{G}$ be a subspace of $\mathcal{H}$. For $A \in \mathcal{L}(\mathcal{H})$, denote $A_{\mathcal{G}} := \Pi_\mathcal{G} A \Pi_\mathcal{G} \in \mathcal{L}(\mathcal{G})$, where $\Pi_\mathcal{G}$ is the projector onto $\mathcal{G}$.

Second, as in the classical case, we hope to solve this optimization problem using Lagrange multipliers. With a fixed $\rho$, $S(\sigma\|\rho)$ is a real-valued function of complex matrices. How do we optimize such functions? In principle we could convert everything into real numbers---$M_N(\mathbb{C}) \cong \mathbb{R}^{2N^2}$, so we could view $S(\sigma\|\rho)$ as a function of $2N^2$ real parameters and implement conventional optimization methods. However, this conversion is generally tedious, and the resulting expression for $S(\sigma\|\rho)$ cumbersome. The `Wirtinger Calculus' provides a relatively simple methodology for the optimization of such functions, through the use of `Wirtinger derivatives'. We state the main definitions and results of this framework in Appendix \ref{Appendix:Wirtinger_Calculus}.

We have the following result, which partially resolves Problem \ref{Problem:quantum_motivating_problem_statement}:
\begin{theorem}\label{Theorem:partial_solution_to_initial_quantum_problem}
The solution to Problem \ref{Problem:quantum_motivating_problem_statement} takes the form
\begin{align}\label{Equation:Solution_form_decomposition}
\sigma^\star = \sigma_{\supp \rho}^\star \oplus \sigma_{\ker \rho}^\star,
\end{align}
where
\begin{align}\label{Equation:Solution_form}
\sigma_{\supp \rho}^\star = \frac{e^{\lambda^\star \cdot H_{\supp \rho} + \log \rho_{\supp \rho}}}{\tr(e^{\lambda^\star \cdot H_{\supp \rho} + \log \rho_{\supp \rho}})} \qquad \text{and} \qquad \sigma_{\ker \rho}^\star = \mathbf{0}.
\end{align}\label{Equation:Getting_lambda}
The optimal values $\lambda^\star \in \mathbb{R}^d$ are to be determined from the constraints 
\begin{align}
\tr\left( e^{\lambda^\star\cdot (H_{\supp \rho}-m) + \log \rho_{\supp \rho}}(H_{i,\supp \rho}-m_i) \right) = 0, i \in [d].
\end{align}
\end{theorem}

The proof of Theorem \ref{Theorem:partial_solution_to_initial_quantum_problem} draws inspiration from the classical version but additional technical hurdles need to be overcome, as mentioned above. The details are elaborated in Appendix \ref{Appendix:quantum_solution}.

Motivated by the form of the state $\sigma_{\supp \rho}^\star$ in Theorem \ref{Theorem:partial_solution_to_initial_quantum_problem}, we make the following definition:

\begin{definition}[Quantum Esscher Transform]\label{Definition:quantum_esscher_transform}
Given a density operator $0 < \rho \in \mathcal{D}(\mathcal{H})$, observables $H_i,\;i \in [d]$ and $\theta \in \mathbb{R}^d$. The density operator
\begin{align*}
\rho_{\theta,H} := \frac{e^{\theta\cdot H + \log \rho}}{\tr(e^{\theta\cdot H + \log \rho})}
\end{align*}
is called the $(\theta,H)$-quantum Esscher Transform of $\rho$.
\end{definition}

\begin{remark}
The state $\sigma_{\supp \rho}^\star$ in Theorem \ref{Theorem:partial_solution_to_initial_quantum_problem} is thus a $(\lambda^\star,H_{\supp \rho})$-quantum Esscher Transform of $\rho_{\supp \rho} > 0$. Also note that the quantum Esscher Transform subsumes the classical Esscher Transform as a special case, wherein $\rho,H_i$ are diagonal and thus commute.
\end{remark}

\textit{Connection to quantum imaginary time evolution}\label{Section:Connection_ITE} $-$
Quantum imaginary-time evolution (QITE) is a conceptual tool used to find ground states of Hamiltonians \cite{mcardle2019variational, motta2020determining}. From the real-time Schr\"{o}dinger equation one obtains the imaginary-time version $\frac{\partial |\psi\rangle}{\partial \tau}=-H |\psi\rangle $ by performing a Wick rotation, i.e. setting $\tau=it$. For general mixed states $\rho$, the imaginary-time Liouville-von Neumann equation \cite{berman1991time} is given by
\be\label{Equation:LvN_equation}
\frac{\partial \rho}{\partial \tau}=-\{H,\rho\}+2\langle H\rangle \rho,
\ee 
from which the solution is derived as
\be\label{Equation:QITE_rho}
\rho(\tau)= A(\tau) e^{-\tau H} \rho(0) e^{-\tau H},
\ee
where $A(\tau)=1/\tr(e^{-2\tau H}\rho(0))$ is the normalisation factor. 

In \cite{olivares2007quantum} it was asserted that under certain conditions, namely `when the prior and posterior states are close to each other with respect to the Fisher information metric', the minimizing relative entropy problem could be solved by formally integrating a `quantum trajectory' equation \cite{olivares2007quantum,braunstein1996geometry}. This equation takes on the same form as Eq. \ref{Equation:LvN_equation}, and thus its solution is given by Eq. \ref{Equation:QITE_rho}. 
More specifically, we have
\begin{align*}
    \rho(\theta) = \frac{e^{\theta \cdot H /2}\rho e^{\theta \cdot H /2}}{\tr(e^{\theta \cdot H} \rho)},
\end{align*}
where $\theta$ are the Lagrange multipliers. Here we simply observe that $\rho(\theta)$ resembles the imaginary-time-evolved state in Eq.~(\ref{Equation:QITE_rho}) if $\theta$ is one-dimensional and after making the substitution $\tau=-\theta/2$. Since the quantum Esscher Transform provides an exact solution to Problem \ref{Problem:quantum_motivating_problem_statement}, under the aforementioned conditions presumed by \cite{olivares2007quantum}, we note the connection between the quantum Esscher Transform and QITE.
\begin{figure}
    \centering
    \includegraphics[width=80mm]{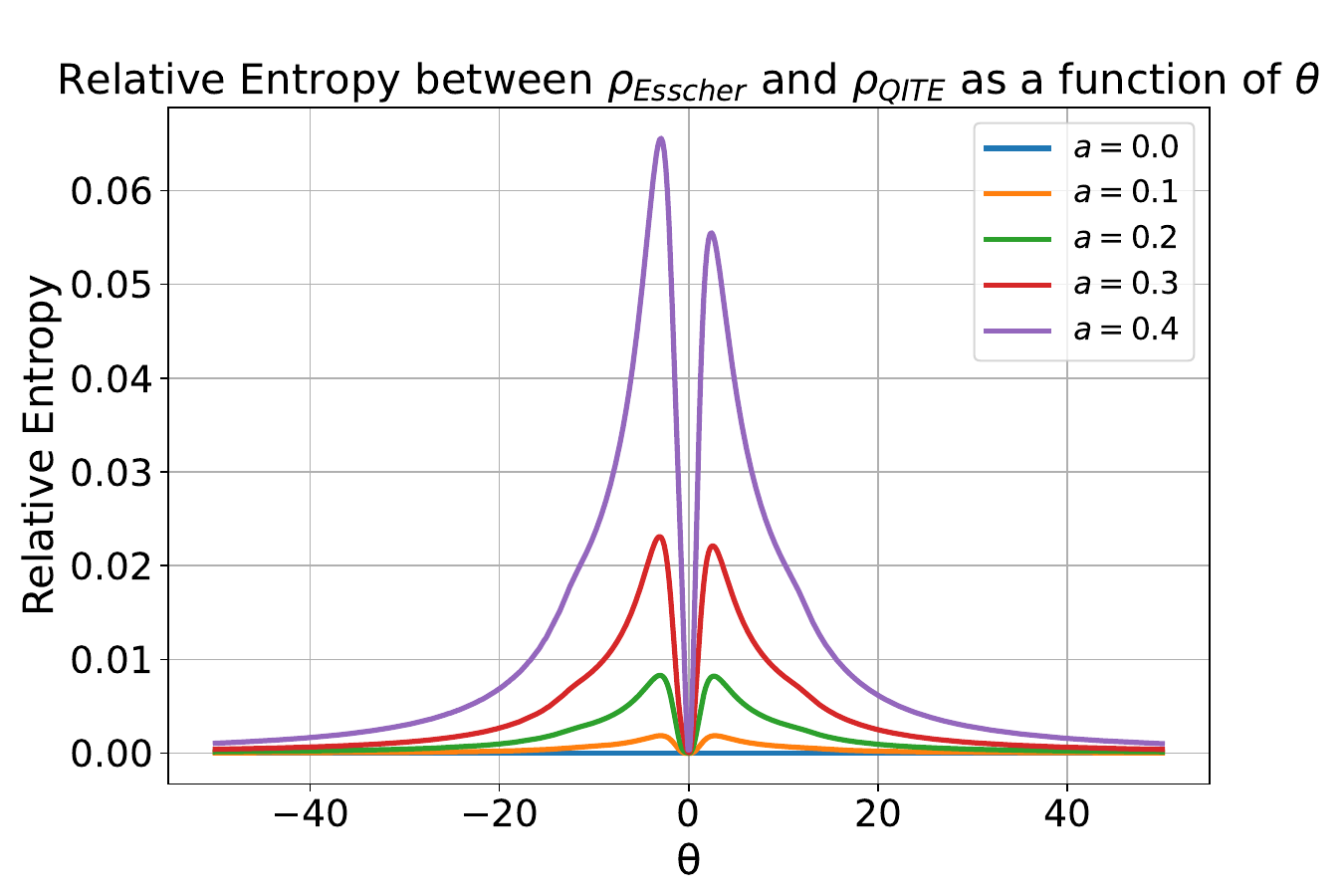}
    \caption{Relative entropy between the quantum Esscher transformed state $\rho_{\text{Esscher}}$ and the imaginary-time evolved $\rho_{\text{QITE}}$ as a function of the parameter $\theta$. Consider a single qubit system, with initial positive state $\rho = 0.5(|0\rangle\langle 0|+|1\rangle\langle 1|)+a (|0\rangle\langle 1| + |1\rangle\langle 0|)$, where we vary $a$ from $0.0$ to $0.4$. We also choose the Hamiltonian $H = \sigma_x + \sigma_z$. The plot above illustrates how the distinguishability of these two states varies with respect to $\theta$.}
     \label{fig:relative_entropy}
\end{figure}

In Fig. \ref{fig:relative_entropy} we provide a simple plot illustrating the difference between $\rho_{\text{Esscher}}$ and the $\rho_{\text{QITE}}$ for a simple initial state $\rho = 0.5(|0\rangle\langle 0|+|1\rangle\langle 1|)+a (|0\rangle\langle 1| + |1\rangle\langle 0|)$, where $a \in[0,0.5)$ is the coherence parameter, and Hamiltonian $H = \sigma_x + \sigma_z$. We use relative entropy as the measure of distinguishability between the two states. We note that as $\theta \rightarrow 0$, the relative entropy vanishes, as is immediate from the definitions. We also note that as $\theta \rightarrow \pm \infty$ the relative entropy vanishes -- because the $\theta \cdot H$ terms then dominate the $\log \rho$ term. In intermediate ranges for $\theta$, the relative entropy is nonzero, and this is interpreted as the error arising when one uses $\rho_{\text{QITE}}$ as a proxy for $\rho_{\text{Esscher}}$. For the dependency on $a$, since $a$ parametrizes the noncommutativity between $\rho$ and $H$, when $a \to 0$ the relative entropy remains small as $\theta$ varies. The relative entropy becomes maximal when $a \rightarrow 0.5$.

Next, we discuss how to implement the quantum Esscher Transform on quantum computers using modern techniques based on block-encodings (BE) and the quantum singular value transformation (QSVT). The relevant tools and techniques of the framework are collated in the appendix.

\section{Implementation on quantum computers}
In this section, we provide a quantum algorithm implementing the quantum Esscher Transform, based on block-encodings and QSVT. We assume the inputs come in the form of block-encodings. Our algorithm outputs the Esscher-transformed state in block-encoded form (and subsequent translations to the physical state itself).

Reference \cite{gilyen2019quantum} demonstrates how to construct block-encodings for density operators $\rho$ within the purified quantum query-access model (see Definition \ref{Definition:Purified_access} and Proposition \ref{Proposition:block_encoding_density_operators} below). For the Hermitian operators $H_i$ which are generally not density operators, their block-encodings can be constructed efficiently for many physical Hamiltonians, or if the $H_i$'s are stored in sparse data structures or KP trees. Along the way we shall also need as an auxiliary tool `state-preparation pairs' (see Definition \ref{Definition:State_preparation_pair}), to prepare linear combinations of the Hamiltonians. We assume immediate access to these, as we do for block-encodings. For the construction of state-preparation pairs, one can refer to \cite{van2018improvements}.

\subsection{Technical lemmas}
The  logarithm of the density matrix $\rho$ is a key ingredient of the quantum Esscher Transform. Here we provide a technical lemma on constructing a block-encoding of the logarithm of a density matrix from the block-encoding of that matrix.

\begin{lemma}[Block-encoding of $\log \rho$]\label{Lemma:step1}
Given $U_\rho$, a $(1,a,0)$-BE of an $n$-qubit density operator $\frac{1}{\kappa} \leq \rho \leq 1$, where $\kappa>1$, and polynomial approximation error tolerance $\varepsilon_\text{poly}>0$. Then we have a $\left( 2(1+\log 2\kappa),\;a+2,\;\varepsilon_\text{poly} \right)$-BE of $\log \rho$, the construction of which makes $\mathcal{O}\left( \kappa \log \left( \frac{\log \kappa}{\varepsilon_{\text{poly}}} \right)\right)$ queries to $U_\rho$.
\end{lemma}
\begin{proof}
First we construct a polynomial approximation of $\log x$. More specifically, we check that the function $\log x$ satisfies the conditions of Proposition \ref{Proposition:gslw_polynomialapproximation}, with the appropriate $x_0,r,\delta$ and $B$. Corollary \ref{Corollary:BE_of_hermitian_matrices} then gives us the desired block-encoding.

The following derivation is based on the proof of Corollary 67, \cite{gilyen2019quantum} and Lemma 11, \cite{gilyen2019distributional}. Negative power functions $x^{-c}$ share with $\log x$ the common property of going to infinity as $x$ approaches $0$, thus the Taylor expansions of these functions are performed about $x=1$. Choose $x_0=1$, $r=1-\frac{1}{\kappa}$ and $\delta=\frac{1}{2\kappa}$. The Taylor series of $\log x$ about $x=1$ is $\log x = \sum_{k=1}^\infty \frac{(-1)^{k+1}}{k}(x-1)^k$. With $a_k= \frac{(-1)^{k+1}}{k}$, the series-of-coefficients bound $B$ in Proposition \ref{Proposition:gslw_polynomialapproximation} is
\begin{eqnarray}
\sum_{k=1}^\infty (r+\delta)^k|a_k| &=& \sum_{k=1}^\infty \frac{(1-\frac{1}{2\kappa})^k}{k} = \sum_{k=1}^\infty \frac{(-1)^k}{k} \left( \frac{1}{2\kappa}-1 \right)^k \nonumber\\
&=& -\log \frac{1}{2\kappa} = \log 2\kappa =: B.
\end{eqnarray}
Corollary \ref{Corollary:BE_of_hermitian_matrices} gives us the unitary $U_{\log \rho}$, which is a $\left( 2(1+\log 2\kappa),\;a+2,\; \varepsilon_\text{poly} \right)$-encoding of $\log \rho$, which can be constructed using $\mathcal{O}\left( \kappa \log \left( \frac{\log \kappa}{\varepsilon_{\text{poly}}} \right)\right)$ queries to $U_\rho$.
\end{proof}
Next, we provide a lemma to construct the block-encoding of an exponentiated matrix from the block-encoding of that matrix.

\begin{lemma}[Block-encoding of $e^H$]\label{Lemma:step3}
Given $U_H$, a $(\alpha,a,\varepsilon)$-BE of $H$ and polynomial approximation error tolerance $\varepsilon_\text{poly}>0$, there is a
$\left( 4,\;a+2,\;\varepsilon_\text{poly} + 16t\sqrt{\varepsilon/\alpha} \right)$-BE of $e^H/e^\alpha$, constructible using $t$ queries to $U_H$. Here
\begin{align*}
t = \mathcal{O}\left(\sqrt{\max (\alpha,\log \frac{1}{\varepsilon_\text{poly}}) \log \frac{1}{\varepsilon_\text{poly}}}\right).
\end{align*}
\end{lemma}
\begin{proof}
By Corollary 64, \cite{gilyen2019quantum}, there exists $P \in \mathbb{R}[x]$ of degree $t = \mathcal{O}\left(\sqrt{\max (\alpha,\log \frac{1}{\varepsilon_{\text{poly}}}) \log \frac{1}{\varepsilon_{\text{poly}}}}\right)$ such that $\|\frac{e^{\alpha x}}{e^\alpha} - P(x)\|_{[-1,1]} \leq \varepsilon_{\text{poly}}$. Furthermore $\|P(x)\| \leq \|\frac{e^{\alpha x}}{e^\alpha} - P(x)\|_{[-1,1]}+\|\frac{e^{\alpha x}}{e^\alpha}\|_{[-1,1]} \leq 1+B$, where $B=1$. Applying Corollary \ref{Corollary:BE_of_hermitian_matrices} with $f(x) = \frac{e^x}{e^\alpha}$ gives a $\left( 4,\;a+2,\;\varepsilon_\text{poly} + 16t\sqrt{\varepsilon/\alpha} \right)$-encoding of $e^H/e^\alpha$, making $t$ queries to $U_H$.
\end{proof}

\newpage

\onecolumngrid

\begin{algorithm}[H]\label{Algorithm:QET}
\caption{Quantum Esscher Transform via QSVT -- \textsc{QEsscher($\rho,H,\theta$)}}
\begin{algorithmic}[1]
\Require 
\Statex - Unitary $O_\rho$ preparing the purification of the $n$-qubit density operator $\frac{1}{\kappa} \leq \rho \leq 1$ using $n_\rho$ ancillary qubits
\Statex - Quantum circuits $U_j$ which are $(1,a,\varepsilon_{\text{BE}})$-BEs of $H_j$ for $j \in [d]$, where $\varepsilon_\text{BE} = \left( \frac{\varepsilon}{8 \log \frac{1}{\varepsilon}} \right)^2$
\Statex - Parameters $\theta \in \mathbb{R}^d$
\Statex - Output block-encoding error $0 < \varepsilon < 2^{-\|\theta\|_1-2(1+\log 2\kappa)}$.
\Ensure A $(1,\;\max\{a,n+n_\rho\}+\lceil \log d \rceil+4,\;\varepsilon)$-BE of $$\sigma = \frac{e^{\sum_i \theta_iH_i + \log \rho}}{\mathcal{N}},$$ where $\mathcal{N} = e^{\|\theta\|_1 + 2(1 + \log 2\kappa)}$ is a subnormalization factor.
\State Use $O_\rho$ to construct $U_\rho$, a a $(1,n+n_\rho,0)$-BE of $\rho$.
\State Construct $U_{\log \rho}$, a $(2(1+\log 2\kappa),\;n+n_\rho+2,\;\varepsilon_{\text{BE}})$-BE of $\log \rho$. This makes $t = \mathcal{O}\left( \kappa \log \left( \frac{\log \kappa}{\varepsilon_{\text{BE}}} \right)\right)$ queries to $U_\rho$, see Lemma \ref{Lemma:step1}.
\State Construct the $(\beta,b,\varepsilon_{\text{SP}})$-state-preparation-pair $(P_L,P_R)$ for $\alpha \odot \theta$, where
\Statex $\beta \gets \|\theta\|_1 + 2(1+\log 2\kappa)$
\Statex $b \gets \lceil \log d \rceil$
\Statex $\varepsilon_\text{SP} \gets \beta \varepsilon_\text{BE}$
\State Using $(P_L,P_R)$, combine $U_{\log \rho}$ and $U_j$, $j \in [d]$ to give $U_H$, a $(\beta,\;\max\{a,n+n_\rho\}+2+\lceil \log d \rceil,\;2\beta\varepsilon_\text{BE})$-BE of $H:=\sum_i \theta_iH_i + \log \rho$. This makes 1 query to $(P_L,P_R)$ and 1 query to $U_{\log \rho}$ and each $U_j$, see Proposition \ref{Proposition:General_linear_combination_BE}.
\State Construct $U_\sigma$, a $(1,\;\max\{a,n+n_\rho\}+4+\lceil \log d \rceil,\;\varepsilon)$-BE of $\sigma:= e^H/\mathcal{N}$. Makes $t = \mathcal{O}\left(\log \frac{1}{\varepsilon}\right)$ queries to $U_H$, see Lemma \ref{Lemma:step3}.
\State \Return $U_\sigma$.
\end{algorithmic}
\end{algorithm}

\twocolumngrid

\subsection{Algorithm}
We now provide the algorithm implementing the quantum Esscher Transform, see Algorithm \ref{Algorithm:QET}. We specify the constraints on the inputs and the guarantees on the output in the algorithm itself. A step-by-step analysis of Algorithm \ref{Algorithm:QET} is provided below in detail, whereafter the overall (query) complexity is stated. We summarize these information in Theorem \ref{Theorem:Theorem_for_algorithm}.

\begin{theorem}\label{Theorem:Theorem_for_algorithm}
Let us be given the block-encodings of $\rho$ and $H_j$, $j \in [d]$, parameters $\theta \in \mathbb{R}^d$ and error tolerance $\varepsilon$ as specified in Algorithm \ref{Algorithm:QET}. Then Algorithm \ref{Algorithm:QET} outputs an $\varepsilon$-approximate block-encoding of the (subnormalized) quantum Esscher Transform 
$
\sigma = \frac{e^{\sum_i \theta_iH_i + \log \rho}}{\mathcal{N}},
$
making 
$$
\widetilde{\mathcal{O}}\left(\kappa \log^2 \left( \frac{1}{\varepsilon} \right)\right)
$$ queries to $U_\rho$ and 
$$
\mathcal{O}\left(\log \frac{1}{\varepsilon}\right)
$$ queries to each $U_j$.
\end{theorem}

\begin{proof}[Proof of Theorem \ref{Theorem:Theorem_for_algorithm}]
Now we analyze the steps of Algorithm \ref{Algorithm:QET} in more detail to give the query complexity of \textsc{QEsscher($\rho,H,\theta$)}.

\begin{description}
\item[\textbf{Step 1}] From Proposition \ref{Proposition:block_encoding_density_operators} we construct $U_\rho = \widetilde{O_\rho} := (O_\rho^\dag \otimes I_n)(I_{n+n_\rho} \otimes \text{SWAP}_n)(O_\rho \otimes I_n)$, a $(1,n+n_\rho,0)$-BE of $\rho$. This makes $\mathcal{O}(1)$ queries to $O_\rho$.

\item[\textbf{Step 2}] This step entails a polynomial approximation to the logarithm function on the interval $[\frac{1}{\kappa},1]$. Denote by $\varepsilon_{\text{poly}}$ the approximation error tolerance. Choose $\varepsilon_{\text{poly}} \leq \varepsilon_{\text{BE}}$. Lemma \ref{Lemma:step1} gives $U_{\log \rho}$, a $(2(1+\log 2\kappa),\;n+n_\rho+2,\;\varepsilon_{\text{BE}})$-BE of $\log \rho$. The construction of $U_{\log \rho}$ makes $t = \mathcal{O}\left( \kappa \log \left( \frac{\log \kappa}{\varepsilon_{\text{BE}}} \right)\right)$ queries to $U_\rho$, where $t$ is the degree of the approximating polynomial (see Proposition \ref{Proposition:gslw_polynomialapproximation}/Corollary \ref{Corollary:BE_of_hermitian_matrices}).

\item[\textbf{Step 3}] Construct a $(\beta,b,\varepsilon_\text{SP})$-state-preparation-pair $(P_L,P_R)$ for $\alpha \odot \theta \in \mathbb{R}^{d+1}$, where $\alpha = (1^d, 2(1+\log 2\kappa))$ and $\theta = (\theta_1,\dots,\theta_d,1)$ (see Proposition \ref{Proposition:General_linear_combination_BE}). Choose $\beta = \|\alpha \odot \theta \|_1 = \|\theta\|_1 + 2(1+\log 2\kappa)$. $b$ has to be such that $d+1 \leq 2^b$, so choose $b = \lceil \log d \rceil$. Finally, choose $\varepsilon_\text{SP} \leq \beta\varepsilon_\text{BE}$. The construction of $(P_L,P_R)$ can be achieved using $\mathcal{O}(d)$ elementary gates \cite{berry2015simulating}.

\item[\textbf{Step 4}] Now we make use of our access to the state-preparation-pair $(P_L,P_R)$. To form linear combinations of block-encodings, the number of ancilla qubits required for each constituent block-encoding should be the same, see Proposition \ref{Proposition:GSLW_linear_combination_BE}/\ref{Proposition:General_linear_combination_BE}.  Remark \ref{Remark:Block_encoding} shows that we can always equalize this number of ancilla qubits by padding with additional ancillas. The equalized number of ancillas is $\max\{a,n+n_\rho+2\} \leq \max\{a,n+n_\rho\}+2$. We could also  take $a+n+n_\rho+2$, but we want to minimize the number of ancilla qubits. From Proposition \ref{Proposition:General_linear_combination_BE} we get $U_H$, a $(\beta,\;\max\{a,n+n_\rho\}+2+\lceil \log d \rceil,\;2\beta\varepsilon_\text{BE})$-BE of $H := \sum_i \theta_i H_i + \log \rho$, making 1 query to $(P_L,P_R)$ and 1 query to $U_{\log \rho}$ and each $U_j$.

\item[\textbf{Step 5}] Finally, we construct a block-encoding for $e^H/\mathcal{N}$. At this stage, we have a $(\beta,\;\max\{a,n+n_\rho\}+2+\lceil \log d \rceil,\;2\beta\varepsilon_\text{BE})$-BE of $H$. Lemma \ref{Lemma:step3} gives a $(1,\;\max\{a,n+n_\rho\}+\lceil \log d \rceil+4,\;\varepsilon_\text{poly}/4 + 4t\sqrt{2\varepsilon_\text{BE}})$-BE of $\sigma = e^H/{4e^\beta}$ (thus $\mathcal{N}=4e^\beta$), where $t = \mathcal{O}\left(\sqrt{\max (\beta,\log \frac{1}{\varepsilon_{\text{poly}}}) \log \frac{1}{\varepsilon_{\text{poly}}}}\right)$. It remains to make judicious choices for $\varepsilon_\text{poly}$ (note that the $\varepsilon_\text{poly}$ at this step need not be the same as the one in Step 2) and $\varepsilon_\text{BE}$ in order to ensure the overall block-encoding error is less than $\varepsilon$, i.e. 
\begin{align}\label{Equation:Final_epsilon_choices}
\frac{\varepsilon_\text{poly}}{4} + 4t\sqrt{2\varepsilon_\text{BE}} \leq \varepsilon.
\end{align}
Now given a sufficently small $\varepsilon$ such that $\varepsilon \leq 2^{-\beta}$, choose $\varepsilon_\text{poly} = \min\{\varepsilon,2^{-\beta}\} = \varepsilon$ and
$$
\varepsilon_\text{BE} = \left( \frac{\varepsilon}{8 \log \frac{1}{\varepsilon}} \right)^2.
$$
These choices ensure Equation $\ref{Equation:Final_epsilon_choices}$ is satisfied.
Note that $\lim_{x \rightarrow 0} \frac{x}{\log \frac{1}{x}} = 0$, so $\varepsilon_\text{BE} \rightarrow 0$ as $\varepsilon \rightarrow 0$. The degree of the approximating polynomial, and thus the number of queries to $U_H$ required, is $t = \mathcal{O}\left(\sqrt{\max (\beta,\log \frac{1}{\varepsilon_\text{poly}}) \log \frac{1}{\varepsilon_\text{poly}}}\right) = \mathcal{O}\left(\log \frac{1}{\varepsilon}\right)$. Recall that constructing $U_H$ itself makes 1 query to $U_{\log \rho}$ and each $U_j$. Lastly, observe that $\|e^H\| \leq e^{\|H\|} \leq e^{\sum_i |\theta_i| + \log \kappa} \leq e^\beta < \mathcal{N}$, so $\mathcal{N}$ is a valid subnormalization factor.
\end{description}

\noindent \textbf{Overall complexity:} $U_\sigma$ makes $\mathcal{O}(\log \frac{1}{\varepsilon})$ queries to $U_H$. $U_H$ queries $U_{\log \rho}$ and each $U_j$ exactly once, and $U_{\log \rho}$ in turn makes $\mathcal{O}\left( \kappa \log \left( \frac{\log \kappa}{\varepsilon_{\text{BE}}} \right)\right)$ queries to $U_\rho$. Accordingly, the implementation of $U_\sigma$ makes 
\begin{align*}
&\mathcal{O}\left(\log \frac{1}{\varepsilon}\right) \cdot \mathcal{O}\left(\kappa \log \left( \log \kappa \cdot \frac{1}{\varepsilon^2} \cdot \log^2 \frac{1}{\varepsilon} \right)\right) \\\nonumber
&\subseteq \mathcal{O}\left(\kappa \log \left( \frac{\log \kappa}{\varepsilon} \right) \log \left( \frac{1}{\varepsilon} \right)\right) \subseteq \widetilde{\mathcal{O}}\left(\kappa \log^2 \left( \frac{1}{\varepsilon} \right)\right)
\end{align*}
queries to $U_\rho$ and
$\mathcal{O}\left(\log \frac{1}{\varepsilon}\right)$ queries to each $U_j$, thus
$$
    \mathcal{O}\left(d \log \frac{1}{\varepsilon}\right)
$$
queries to $\{U_j\}_{j=1}^d$, the constraint operators collectively considered. 
\end{proof}

\subsection{Further discussion}
If the positive definite $\rho \in \mathbb{C}^{N \times N}$ is full rank, the condition number is $\kappa \geq N$ since the eigenvalue lower bound $\frac{1}{\kappa}$ must be $\leq 1/N$. 
Then the $U_\rho$-query complexity grows at least linearly with $N$.
Hence, our Esscher Transform is most relevant for low-rank cases. Assume we have $r$ non-zero eigenvalues $\geq 1/\kappa$. As a consequence $r\leq \kappa$ holds. While the condition number can still be exponential if the smallest eigenvalue is exponentially small, when the smallest eigenvalue is $1/{\rm poly} (r)$, we obtain a well-behaved query complexity. 
In addition we can allow for smaller eigenvalues, especially when we are interested only in low-rank approximations of the Esscher Transform. Let $1/\kappa_{\rm eff} \geq 1/\kappa$, with the effective condition number $\kappa_{\rm eff}$. With slight adaptations, our method can implement the Esscher Transform on the effectively well-conditioned subspace, while leaving the other part undefined. This incurs an error compared to the full Esscher Transform proportional to the importance of the neglected eigenvalues, but may be acceptable in many practical situations. 
Recall that low-rank approximations are frequently performed in statistics and machine learning. 

If the desired output model is a normalized state, one can apply similar techniques for Gibbs sampling to extract the normalized Esscher-Transformed state from the output of Algorithm \ref{Algorithm:QET}. We briefly describe this procedure and the overhead cost it incurs. More details can be found in Chapter 3 of \cite{gilyen2019thesis}. Let $\varepsilon>0$ denote the desired precision in trace distance between our approximate output and the ideal state. First, we prepare a maximally entangled state on two registers. Use Algorithm \ref{Algorithm:QET} to construct a $1$-block-encoding $U$ of $ e^{\frac{\sum_i \theta_iH_i + \log \rho}{2}}/\sqrt{\mathcal N}$ where $\mathcal{N} = e^{\|\theta\|_1 + 2(1 + \log 2\kappa)}$, with block-encoding error 
$0 < \varepsilon_1 < \varepsilon/N^2$. Then apply $U$ to the second register to obtain a state $|\psi\rangle$, so that tracing out the first register yields an approximate subnormalized state with trace distance error of $\mathcal{O}\left(\varepsilon/N\right)$. That is,
\begin{align*}
    \left\|\tr_1 (\bra{0} \otimes I) \ket{\psi} \bra{\psi} (\ket{0} \otimes I) - \frac{e^{\sum_i \theta_iH_i + \log \rho}}{N \mathcal N}\right\|_T = \mathcal{O}\left(\frac{\varepsilon}{N}\right).
\end{align*}
With $\mathcal Z := \tr \left(e^{\sum_i \theta_iH_i + \log \rho}\right)$, this state, when postselected after $\mathcal{O}\left(\sqrt{\frac{N \mathcal N}{\mathcal Z}}\log \frac{1}{\varepsilon}\right)$ steps of fixed-point amplitude amplification (refer to Theorem 27 in \cite{gilyen2019quantum}), results in a density operator 
$\varepsilon$-close to the normalized Esscher-Transformed state
\begin{align*}
\frac{e^{\sum_i \theta_iH_i + \log \rho}}{\tr (e^{\sum_i \theta_iH_i + \log \rho})}
\end{align*}
in trace distance. Taking this overhead cost into account and assuming $\varepsilon$ is sufficiently small (such that the block-encoding error satisfies $\varepsilon_1 < 2^{-\|\theta\|_1-2(1+\log 2\kappa)}$), the total query complexity of preparing the approximate Esscher-Transformed state is
\begin{align*}
&\widetilde{\mathcal{O}}\left(\kappa \log^2 \left( \frac{N^2}{\varepsilon} \right)\right) \cdot \mathcal{O}\left(\sqrt{\frac{N \mathcal N}{\mathcal Z}}\log \frac{1}{\varepsilon}\right) \\\nonumber
&\subseteq \widetilde{\mathcal{O}}\left(\kappa\sqrt{\frac{N \mathcal N}{\mathcal Z}} \log^3 \left( \frac{1}{\varepsilon} \right)\right).
\end{align*}

\section{Conclusion}
In this paper, we considered a minimum relative entropy problem for the density operator subject to equality constraints. We formally solved this problem and the solution form inspired us to define the Quantum Esscher Transform (QUEST), a generalization of the classical Esscher Transform to the quantum setting. We discussed its implementation on fault-tolerant quantum computers, leveraging techniques based on the QSVT framework. Given as inputs block-encodings of the initial quantum state and the constraint operators, the algorithm outputs an $\varepsilon$-approximate block-encoding of the Esscher-Transformed state with $U_\rho$-query complexity
\begin{align*}
    \mathcal{O}\left(\kappa \log \left( \frac{\log \kappa}{\varepsilon} \right) \log \left( \frac{1}{\varepsilon} \right)\right) \subseteq \widetilde{\mathcal{O}}\left(\kappa \log^2 \left( \frac{1}{\varepsilon} \right)\right)
\end{align*}
and $\{U_j: j \in [d]\}$-query complexity
$$
    \mathcal{O}\left(d \log \frac{1}{\varepsilon}\right).
$$
Several avenues remain open for future work:
\begin{itemize}
\item Is there a quantum algorithmic framework that can fully solve the minimum relative entropy problem? Our current approach only presents the formal solution for the optimal parameter $\lambda^*$. Approaches such as Newton's algorithm with backtracking was suggested in \cite{zorzi2013minimum}, the quantized version of which could be studied. Additionally, \cite{anshu2020sample} demonstrated that $\lambda^*$ can, in principle, be found with a convex optimization program. Can we design a quantum algorithm to effectively address this problem?

\item One could explore strategies for alternative input models. Our current work exclusively considered the purified access model, wherein the preparation of the purification of the input state was assumed. In contrast, the sampling access model, which assumes multiple independent copies of the input state, is another commonly used model. Gily{\'e}n et al. \cite{gilyen2022improved} has proposed an approach to implement approximate block-encodings of $\rho$, starting with sample access. This approach is based on a combination of density matrix exponentiation \cite{lloyd2014quantum, kimmel2017hamiltonian} and QSVT, and allows us to implement the quantum Esscher Transform in the sampling access model. We leave the total cost of this procedure for further analysis.

\item In Section \ref{Section:Connection_ITE}, we noted potential connections between the quantum Esscher Transform and imaginary-time evolution. To give these substance, further investigation is required.

\item Various applications could be envisioned for the quantum Esscher Transform. Its classical version has found usage for numerous problems in domains such as statistics, machine learning, and finance. These problems have quantum analogues, which could benefit from the quantum Esscher Transform and its implementation on quantum computers.

\end{itemize}

\section*{Acknowledgments}
The authors would like to thank Po-Wei Huang, Xiufan Li, Zhan Yu, Roberto Rubboli and Serge Massar for helpful discussions. This work is supported by the National Research Foundation, Singapore, and A*STAR under its CQT Bridging Grant and its Quantum Engineering Programme under grant NRF2021-QEP2-02-P05. KK acknowledges support from Leong Chuan Kwek, under project grant R-710-000-007-135.

\onecolumngrid
\bibliographystyle{alpha}
\bibliography{QEsscher}

\newpage

\begin{appendix}
\section{Proof of Theorem \ref{Theorem:solution_to_classical_problem}}\label{Appendix:Complete_solution_lambda}
Before delving into the proof, we introduce some notation and state a lemma to facilitate its presentation.
The exponential family of $P$ with respect to the random variable $X$ is the set of measures
\begin{align*}
\Lambda = \left\{\frac{e^{\lambda\cdot X}P}{\mathbb{E}_P[e^{\lambda \cdot X}]}: \lambda \in \mathbb{R}^d \right\}.
\end{align*}
Also, let
\begin{align*}
M = \{Q: \mathbb{E}_Q[X] = m \}.
\end{align*}

\begin{lemma}(Proposition 3.24 -- \cite{follmer2011stochastic})
\label{Lemma:minimizing_Q}
Let $P$ be a probability measure on $(\Omega,\Sigma)$ and $X$ be a random variable on $\Omega$. Fix $m \in \mathbbm R^d$. Then for \textit{any} probability measure $Q$ on $(\Omega,\Sigma)$ satisfying $\mathbb{E}_Q[X]= m$, we have 
\be
D(Q\|P) \geq \sup_{\lambda \in \mathbb{R}^d}\left[\lambda \cdot m - \log \mathbb{E}_P[e^{\lambda \cdot X}]\right].
\ee
Moreover the inequality is saturated if $Q = Q_{\lambda'} := e^{\lambda'\cdot X}P/\mathbb{E}_P[e^{\lambda' \cdot X}] \in \Lambda \cap M$ for some $\lambda' \in \mathbb{R}^d$:
\be
D(Q_{\lambda'}\|P) = \lambda' \cdot m - \log \mathbb{E}_P[e^{\lambda' \cdot X}] = \sup_{\lambda\in \mathbb{R}^d} \left[\lambda \cdot m - \log \mathbb{E}_P[e^{\lambda \cdot X}] \right].
\ee
\end{lemma}
\begin{proof}
Each $\lambda \in \mathbb{R}^d$ gives rise to a corresponding $Q_\lambda \in \Lambda$ (note that $Q_\lambda$ need not be in $M$). Then for any arbitrary $Q$, we have
\be
D(Q\|P) &=& \sum_{\omega \in \Omega}Q(\omega)\log \frac{Q(\omega)}{Q_{\lambda}(\omega)} \frac{Q_{\lambda}(\omega)}{P(\omega)}\\\nonumber
&=& D(Q\|Q_{\lambda})+\sum_{\omega \in \Omega}Q(\omega)\log  \frac{Q_{\lambda}(\omega)}{P(\omega)}\\\nonumber
&\geq& \sum_{\omega \in \Omega}Q(\omega)\log  \frac{Q_{\lambda}(\omega)}{P(\omega)}\\\nonumber
&=& \sum_{\omega \in \Omega}Q(\omega)\log \frac{ e^{\lambda \cdot X(\omega)}}{\mathbb{E}_P[e^{\lambda \cdot X}]}\\\nonumber
&=& E_{Q}[\lambda \cdot X]- \log \mathbb{E}_P[e^{\lambda \cdot X}]\\\nonumber
&=& \lambda \cdot m - \log \mathbb{E}_P[e^{\lambda \cdot X}].
\ee
The third inequality is due to Jensen's inequality $D(Q\|P) \geq 0$. Since this holds for all $\lambda \in \mathbb{R}^d$, we conclude that $D(Q\|P) \geq \sup _{\lambda \in \mathbb{R}^d}\left[\lambda \cdot m - \log \mathbb{E}_P[e^{\lambda \cdot X}]\right]$. Furthermore, if $\lambda' \in \mathbb{R}^d$ is such that $Q_{\lambda'} \in \Lambda \cap M$, then letting $Q=Q_{\lambda'}$ and rerunning the same argument sequence above gives
\begin{align*}
D(Q_{\lambda'}\|P) &= \sum_{\omega \in \Omega}Q_{\lambda'}(\omega)\log \frac{Q_{\lambda'}(\omega)}{P(\omega)}\\
&= \sum_{\omega \in \Omega}Q_{\lambda'}(\omega)\log \frac{ e^{\lambda' \cdot X(\omega)}}{\mathbb{E}_P[e^{\lambda' \cdot X}]}\\
&= E_{Q_{\lambda'}}[\lambda' \cdot X]- \log \mathbb{E}_P[e^{\lambda' \cdot X}]\\
&= \lambda' \cdot m - \log \mathbb{E}_P[e^{\lambda' \cdot X}].
\end{align*}
\end{proof}

\begin{proof}[Proof of Theorem \ref{Theorem:solution_to_classical_problem}]
First, we have required $\min_{\omega \in \Omega} X_i(\omega) < m_i < \max_{\omega \in \Omega} X_i(\omega)$ because otherwise the constraints $\mathbb{E}_Q[X_i]=m_i$ cannot be satisfied. The Lagrangian function is
\begin{align*}
\mathcal{L}(Q,\lambda,\eta)= \sum_{\omega}Q(\omega)\log \frac{Q(\omega)}{P(\omega)} - \sum_{i=1}^d \lambda_i \left(\sum_\omega Q(\omega)X_i(\omega)-m_i\right)-\eta\left(\sum_\omega Q(\omega)-1\right).
\end{align*}
Setting the first-order derivatives of $\mathcal{L}(Q,\lambda,\eta)$ with respect to $Q(\omega)$ to zero gives
\begin{align*}
Q^\star(\omega)=\frac{e^{\lambda^\star \cdot X(\omega)}P(\omega)}{\mathbb{E}_P[e^{\lambda^\star \cdot X}]},
\end{align*}
where $\lambda^\star$ is to be determined from the $d$ constraints $\mathbb{E}_Q[X]=m$:
\begin{align}
\mathbb{E}_Q[X]=m &\iff \frac{\mathbb{E}_P[Xe^{\lambda^\star \cdot X}]}{\mathbb{E}_P[e^{\lambda^\star \cdot X}]} - m = 0\\\nonumber
&\iff \frac{\mathbb{E}_P[(X-m)e^{\lambda^\star \cdot (X-m)}]}{\mathbb{E}_P[e^{\lambda^\star \cdot (X-m)}]} = 0\\\nonumber
&\iff \frac{\partial}{\partial \lambda} \log \mathbb{E}_P[e^{\lambda \cdot (X-m)}]|_{\lambda=\lambda^\star} = 0\\\nonumber
&\iff \frac{\partial}{\partial \lambda} \mathbb{E}_P[e^{\lambda \cdot (X-m)}]|_{\lambda=\lambda^\star} = 0.
\end{align}
The last equivalence holds because $\log f(x)$ and $f(x)$ share the same minimum/maximum points, provided $f(x) > 0$ at those points. It remains to show $Q^\star$ indeed minimizes $D(Q\|P)$, subject to the constraints $E_{Q}[X] = m$. But this follows easily from Lemma \ref{Lemma:minimizing_Q}. Furthermore, since $x \mapsto x \log x$ is a strictly convex function, $D(Q\|P)$ is a strictly convex functional of $Q$ and so it can have at most one minimizer in the convex set $M$, thereby showing the uniqueness of $Q^\star$. Finally, again using Lemma \ref{Lemma:minimizing_Q} we have $\lambda^\star = \argmax_{\lambda \in \mathbb{R}^d} \left[\lambda \cdot m - \log \mathbb{E}_P[e^{\lambda \cdot X}] \right] = \argmin_{\lambda \in \mathbb{R}^d} \left[ \log \mathbb{E}_P[e^{\lambda \cdot (X-m)}] \right] = \argmin_{\lambda \in \mathbb{R}^d} \mathbb{E}_P[e^{\lambda \cdot (X-m)}]$. 
\end{proof}

\begin{remark}\label{Remark:Equivalence_classical_ET}
The random variable $X$ induces from the probability measure $P$ the probability mass function $P_X(x) := P(X^{-1}(x))$ on $E:=X(\Omega)$. Assume we have, for probability measures $Q,P$ and random variable $X$, that
\begin{align*}
Q(\omega)=\frac{e^{\theta\cdot X(\omega)}P(\omega)}{\mathbb{E}_P[e^{\theta\cdot X}]}.
\end{align*}
Then for the probability mass functions $Q_X$ and $P_X$ we have
\begin{align*}
Q_X(x) = Q(X^{-1}(x)) &= \sum_{\omega: X(\omega)=x} Q(\omega)\\
&= \frac{\sum_{\omega: X(\omega)=x} e^{\theta\cdot X(\omega)} P(\omega)}{\sum_{\omega \in \Omega} e^{\theta\cdot X(\omega)} P(\omega)}\\
&= \frac{e^{\theta\cdot x} P_X(x)}{\sum_{x \in E}\sum_{\omega: X(\omega)=x} e^{\theta\cdot X(\omega)} P(\omega)}\\
&= \frac{e^{\theta\cdot x} P_X(x)}{\sum_{x \in E} e^{\theta\cdot x} P_X(x)},
\end{align*}
i.e., $Q_X$ is the Esscher Transform of $P_X$ as  defined above. 
\end{remark}

\section{Wirtinger Calculus}\label{Appendix:Wirtinger_Calculus}
The `Wirtinger Calculus' provides a methodology for optimization problems involving complex matrices. It enables `differentiation as usual' with respect to complex matrices. 
In this appendix, we state only the main definitions and results needed to solve Problem \ref{Problem:quantum_motivating_problem_statement}. For a more thorough exposition of this framework, we direct the reader to \cite{koor2023short,hjorungnes2011complex,kreutz2009complex}.

Consider functions of the form $f: \mathbb{C}^{n \times n} \longrightarrow \mathbb{C}$. Since $\mathbb{C}$ is $\mathbb{R}^2$ endowed with the multiplication operation $(a,b) \times (c,d) \mapsto (ac-bd,ad+bc)$, we can view 
\begin{align*}
    f:\; &\mathbb{R}^{2(n \times n)} \longrightarrow \mathbb{R}^2\\
    &(x_{ij},y_{ij})_{i,j \in [n]}=(\mathbf{X},\mathbf{Y}) \mapsto (u(\mathbf{X},\mathbf{Y}),v(\mathbf{X},\mathbf{Y})).
\end{align*} 
For $i=1,\dots,n$ regard $z_{ij},z_{ij}^*$ as functions from $\mathbb{R}^{n \times n} \times \mathbb{R}^{n \times n}$ to $\mathbb{C}$, where
$z_{ij}(\mathbf{X},\mathbf{Y}) = x_{ij}+iy_{ij}$ and $z_{ij}^*(\mathbf{X},\mathbf{Y}) = x_{ij}-iy_{ij}$.\footnote{The notations $z,z^*$ may raise questions on independence. This is irrelevant---one may simply write $z_1,z_2$ if one wishes. We emphasize that (for each $i,j$) the fundamental input variables are the two real numbers $x$ and $y$.} Then we have a function $\tilde{f}: \mathbb{C}^{n \times n} \times \mathbb{C}^{n \times n} \longrightarrow \mathbb{C}$ such that
\begin{align}\label{Equation:tildef_to_f_matrix}
    f(\mathbf{X},\mathbf{Y}) := \underline{\tilde{f}\circ (\mathbf{Z},\mathbf{Z^*})}(\mathbf{X},\mathbf{Y}) = \tilde{f}(\mathbf{Z}(\mathbf{X},\mathbf{Y}),\mathbf{Z^*}(\mathbf{X},\mathbf{Y})) = \tilde{f}(\mathbf{X+iY},\mathbf{X-iY}).
\end{align}
Partial differentiating $f$ with respect to each $x_{ij}$ and $y_{ij}$, and then rearranging terms, we have for $1 \leq i,j \leq n$
\begin{align}\label{Equation:after_rearranging}
    \frac{\partial \tilde{f}}{\partial z_{ij}}(\mathbf{Z}(\mathbf{X},\mathbf{Y}),\mathbf{Z^*}(\mathbf{X},\mathbf{Y})) &= \frac{1}{2}\left(\frac{\partial f}{\partial x_{ij}} - i\frac{\partial f}{\partial y_{ij}}\right) (\mathbf{X},\mathbf{Y})\\\nonumber
    \frac{\partial \tilde{f}}{\partial z_{ij}^*}(\mathbf{Z}(\mathbf{X},\mathbf{Y}),\mathbf{Z^*}(\mathbf{X},\mathbf{Y})) &= \frac{1}{2}\left(\frac{\partial f}{\partial x_{ij}} + i\frac{\partial f}{\partial y_{ij}}\right) (\mathbf{X},\mathbf{Y}).
\end{align}
To preserve the matrix structure of the parameters $z_{ij}$ and $z_{ij}^*$ we use the standard notation
\begin{align}\label{Equation:matrix_wirtinger_derivatives}
    \frac{\partial}{\partial \mathbf{Z}} := 
    \begin{bmatrix}
        \frac{\partial}{\partial z_{11}} & \dots & \frac{\partial}{\partial z_{1n}}\\
        \vdots & \ddots & \vdots\\
        \frac{\partial}{\partial z_{n1}} & \dots & \frac{\partial}{\partial z_{nn}}
    \end{bmatrix} \qquad
    \frac{\partial}{\partial \mathbf{Z^*}} := 
    \begin{bmatrix}
        \frac{\partial}{\partial z_{11}^*} & \dots & \frac{\partial}{\partial z_{1n}^*}\\
        \vdots & \ddots & \vdots\\
        \frac{\partial}{\partial z_{n1}^*} & \dots & \frac{\partial}{\partial z_{nn}^*}
    \end{bmatrix}
\end{align}
and similarly for $\frac{\partial}{\partial \mathbf{X}}$ and $\frac{\partial}{\partial \mathbf{Y}}$. Then Equation \ref{Equation:after_rearranging} is concisely stated as 
\begin{align}
    \frac{\partial \tilde{f}}{\partial \mathbf{Z}}(\mathbf{Z}(\mathbf{X},\mathbf{Y}),\mathbf{Z^*}(\mathbf{X},\mathbf{Y})) &= \frac{1}{2}\left(\frac{\partial f}{\partial \mathbf{X}} - i\frac{\partial f}{\partial \mathbf{Y}}\right) (\mathbf{X},\mathbf{Y})\\\nonumber
    \frac{\partial \tilde{f}}{\partial \mathbf{Z^*}}(\mathbf{Z}(\mathbf{X},\mathbf{Y}),\mathbf{Z^*}(\mathbf{X},\mathbf{Y})) &= \frac{1}{2}\left(\frac{\partial f}{\partial \mathbf{X}} + i\frac{\partial f}{\partial \mathbf{Y}}\right) (\mathbf{X},\mathbf{Y}).
\end{align}
$\frac{\partial}{\partial \mathbf{Z}}$ and $\frac{\partial}{\partial \mathbf{Z^*}}$ are the \textit{matrix Wirtinger derivatives} of $f$. Often, we abuse notation and write both $f(\mathbf{X},\mathbf{Y})$ and $f(\mathbf{Z},\mathbf{Z^*})$, so we can write
\begin{align}
    \frac{\partial}{\partial \mathbf{Z}} = \frac{1}{2}\left(\frac{\partial}{\partial \mathbf{X}} - i\frac{\partial}{\partial \mathbf{Y}}\right), \qquad \frac{\partial}{\partial \mathbf{Z^*}} = \frac{1}{2}\left(\frac{\partial}{\partial \mathbf{X}} + i\frac{\partial}{\partial \mathbf{Y}}\right).
\end{align}

The following three propositions are all we need in this paper. We omit their proofs, which can all be found in \cite{koor2023short}.
\begin{proposition}\label{Proposition:Optimization_with_Wirtinger_derivatives_matrix}
    Let $f: \mathbb{C}^{n \times n} \longrightarrow \mathbb{R}$ be a real-valued function of complex matrices. Then $f$ has a stationary point at $\mathbf{Z}=[z_{ij}]_{i,j \in [n]}$ if and only if
    \begin{align*}
        \frac{\partial f}{\partial \mathbf{Z}}(\mathbf{Z}) = 0 \quad \left( \text{or equivalently}\;\; \frac{\partial f}{\partial \mathbf{Z^*}}(\mathbf{Z}) = 0 \right).
    \end{align*}
\end{proposition}
Whether the solution of the above equation actually gives a minimum/maximum/saddle point has to be checked via additional considerations or by inspecting higher-order derivatives.

\begin{proposition}\label{Proposition:Wirtinger_derivative_of_traces}
Let $\mathbf{Z}$ be a complex, unstructured (see below) matrix and $F(z) = \sum_{n=0}^\infty c_nz^n$ be analytic. Define the scalar function $f(\mathbf{Z,Z^*}):= \tr(F(\mathbf{Z}))$. Then
\begin{align*}
\frac{\partial \tr(F(\mathbf{Z}))}{\partial \mathbf{Z}} = F'(\mathbf{Z})^T
\end{align*}
where $F'(\cdot)$ is the complex derivative of $F(\cdot)$.
\end{proposition}

So far, by writing $f: \mathbb{C}^{n \times n} \longrightarrow \mathbb{C}$ we have implicitly assumed the input matrices have independent components (we call such matrices `unstructured'). This condition often does not hold, e.g. when our matrices of interest are symmetric/Hermitian etc. To obtain the correct Wirtinger derivatives with respect to structured matrices, we resort to the chain rule.
\begin{proposition}[Wirtinger derivatives with respect to Hermitian matrices]\label{Proposition:Wirtinger_hermitian}
    Let $f(\mathbf{Z,Z^*})$ be a function of complex Hermitian matrices. Then the Wirtinger derivatives of $f$ with respect to $\mathbf{Z,Z^*}$ are given by
    \begin{align*}
        \frac{\partial f}{\partial \mathbf{Z}} = \left[\frac{\partial f}{\partial \mathbf{\tilde{Z}}} + \left(\frac{\partial f}{\partial \mathbf{\tilde{Z}^*}}\right)^T\right]_{\mathbf{\tilde{Z}}=\mathbf{Z}} \qquad \text{and} \qquad
        \frac{\partial f}{\partial \mathbf{Z^*}} = \left[\frac{\partial f}{\partial \mathbf{\tilde{Z}^*}} + \left(\frac{\partial f}{\partial \mathbf{\tilde{Z}}}\right)^T\right]_{\mathbf{\tilde{Z}}=\mathbf{Z}}.
    \end{align*}
\end{proposition}
Here, the tildes above $\mathbf{\tilde{Z},\tilde{Z}^*}$ indicate that they are unstructured matrices. Thus, to derive the Wirtinger derivatives with respect to Hermitian matrices, first obtain the Wirtinger derivative of $f$, assuming the inputs are unstructured. Then form the correct expressions given above and reinstate the structured matrices $\mathbf{Z,Z^*}$ as the arguments.

\section{Proof of Theorem \ref{Theorem:partial_solution_to_initial_quantum_problem}}\label{Appendix:quantum_solution}
\begin{proof}
To facilitate the presentation of the solution, certain parts of the argument sequence are collated into lemmas and placed below the main body of this proof.

\textbf{Step 1.} First, for any candidate solution $\sigma$ we enforce $\ker \rho \subseteq \ker \sigma$. By Lemma \ref{Lemma:ker_and_supp_rho}, this implies $\sigma_{\ker \rho} = \mathbf{0}$ and furthermore enables the decomposition of $\sigma$ into a direct sum: $\sigma = \sigma_{\supp \rho} \oplus \sigma_{\ker \rho}$. With this decomposition, we can consider the trace of the operators over just the subspace $\supp \rho$. More specifically, $\tr(\sigma H_i) = \tr(\sigma(\Pi_{\supp \rho} + \Pi_{\ker \rho})H_i(\Pi_{\supp \rho} + \Pi_{\ker \rho})) = \tr(\sigma_{\supp \rho} H_{i,\supp \rho})$\footnote{Recall that for any $A \in \mathcal{L}(\mathcal{H})$, $\ker A \oplus \supp A = \mathcal{H}$, so $\Pi_{\ker A} + \Pi_{\supp A} = I$.} and
\begin{eqnarray}
S(\sigma\|\rho) &=& \tr\{\sigma_{\supp \rho} \oplus \sigma_{\ker \rho}\left(\log (\sigma_{\supp \rho} \oplus \sigma_{\ker \rho}) - \log (\rho_{\supp \rho} \oplus \rho_{\ker \rho})\right)\}\\
&=& \tr\{\sigma_{\supp \rho}(\log \sigma_{\supp \rho} - \log \rho_{\supp \rho})\} + \underbrace{\tr\{\sigma_{\ker \rho}(\log \sigma_{\ker \rho} - \log \rho_{\ker \rho})\}}_{=0}\\
&=& S(\sigma_{\supp \rho} \| \rho_{\supp \rho}).
\end{eqnarray}
Thus, we can replace $\mathcal{H}$ in Problem \ref{Problem:quantum_motivating_problem_statement} by $\supp \rho$, and the operators by their restrictions to $\supp \rho$. Note that $\rho_{\supp \rho}$ is positive definite.

\textbf{Step 2.} Next we obtain the form of $\sigma_{\supp \rho}$. For ease of presentation let us simply denote $(\sigma/\rho/H_i)_{\supp \rho}$ by $(\sigma/\rho/H_i)$. With $\rho$ now positive definite, $\log \rho$ is well-defined. Now we invoke Proposition \ref{Proposition:Optimization_with_Wirtinger_derivatives_matrix} to extract the optimal $\sigma$ by setting $\frac{\partial \mathcal{L}}{\partial \sigma} = \mathbf{0}$.
Set up the Lagrangian
\begin{eqnarray}
\label{Equation:Lagrangian}
\mathcal{L} = \tr\{\sigma(\log \sigma - \log \rho)\} - \sum_i \lambda_i (\tr(\sigma H_i)-m_i) - \eta(\tr\sigma-1)
\end{eqnarray}
where $\lambda_i$ and $\eta$ are the Lagrange multipliers. Making use of Propositions \ref{Proposition:Wirtinger_derivative_of_traces} and \ref{Proposition:Wirtinger_hermitian}, setting $\frac{\partial \mathcal{L}}{\partial \sigma}$ to zero gives
\begin{align*}
\frac{\partial \mathcal{L}}{\partial \sigma} = \mathbf{0} &\implies (\log \sigma)^T + I - (\log \rho)^T - (\lambda \cdot H)^T - \eta I = \mathbf{0}\\
&\implies \sigma = e^{\eta-1}e^{\lambda \cdot H + \log \rho}\\
&\implies \sigma = \frac{e^{\lambda \cdot H + \log \rho}}{\tr(e^{\lambda \cdot H + \log \rho})} \qquad \text{after normalization}.
\end{align*}
It remains to determine $\lambda$ from the constraints $\tr(\sigma H) = m$. Plugging in the above expression for $\sigma$ into the constraints we have
\begin{align*}
\frac{\tr(e^{\lambda \cdot H + \log \rho} H)}{\tr(e^{\lambda \cdot H + \log \rho})} = m \implies &\frac{\tr(e^{\lambda \cdot H + \log \rho} (H-m))}{\tr(e^{\lambda \cdot H + \log \rho})} = 0\\
\implies &\tr(e^{\lambda \cdot (H-m) + \log \rho} (H-m)) = 0.
\end{align*}

\textbf{Step 3.} Now we show that $\sigma^\star$ as given in Eq. \ref{Equation:Solution_form} indeed minimizes $S(\sigma\|\rho)$. But this follows easily from Lemma \ref{Lemma:minimizing_sigma}. Furthermore, since $S(\sigma\|\rho)$ is a strictly convex functional of $\sigma$, it can have at most one minimizer in the convex set $M$, thereby showing the uniqueness of $\sigma^\star$. Finally, again by Lemma \ref{Lemma:minimizing_sigma} we note that $\lambda^\star$ satisfies $\lambda^\star = \argmax_{\lambda \in \mathbb{R}^d} \left[\lambda \cdot m - \log \tr(e^{\lambda \cdot H + \log \rho}) \right] = \argmin_{\lambda \in \mathbb{R}^d} \log \tr(e^{\lambda \cdot (H-m) + \log \rho}) = \argmin_{\lambda \in \mathbb{R}^d} \tr(e^{\lambda \cdot (H-m) + \log \rho})$, where the last equality holds because $\log f(x)$ and $f(x)$ share the same minimum/maximum points, provided $f(x) > 0$ at those points. 
\end{proof}

\begin{lemma}\label{Lemma:ker_and_supp_rho}
Let $\sigma,\rho \in \mathcal{L}(\mathcal{H})$ be normal operators, so that they have spectral decompositions. If $\ker \rho \subseteq \ker \sigma$, then $\sigma_{\ker \rho} = \mathbf{0}$ and $\sigma$ can be partitioned into a direct sum:
\begin{align*}
\sigma = \sigma_{\supp \rho} \oplus \sigma_{\ker \rho}.
\end{align*}
\end{lemma}
\begin{proof}
Expand $\sigma$ in terms of the eigenbasis of $\rho$, $\{\ket{i}\}_{i=0}^{N-1}$. Let $S \subseteq [N]-1$ be the index subset such that $\text{span}\{\ket{i}: i \in S\} = \supp \rho$, so $\text{span}\{\ket{i}: i \in S^c\} = \ker \rho$. We have
\begin{align*}
\sigma =& \sum_{i,j=0}^{N-1} \braket{i|\sigma|j} \ket{i}\bra{j} \\\nonumber
= &\underbrace{\sum_{i \in S}\sum_{j \in S} \braket{i|\sigma|j} \ket{i}\bra{j}}_{=\;\sigma_{\supp \rho}} + \underbrace{\sum_{i \in S}\sum_{j \in S^c} \braket{i|\sigma|j} \ket{i}\bra{j}}_{=\mathbf{0}}\\
&+ \underbrace{\sum_{i \in S^c}\sum_{j \in S} \braket{i|\sigma|j} \ket{i}\bra{j}}_{=\mathbf{0}} + \underbrace{\sum_{i \in S^c}\sum_{j \in S^c} \braket{i|\sigma|j} \ket{i}\bra{j}}_{=\;\sigma_{\ker \rho}=\;\mathbf{0}},
\end{align*}
where the annihilation of the last three terms comes about because for $i \in S^c$, $\ket{i} \in \ker \rho \subseteq \ker \sigma$. 

Note that the partition of an operator into a direct sum over \textit{another} operator's ker and supp subspaces does not hold in general.
\end{proof}

The following lemma is the quantized version of Lemma \ref{Lemma:minimizing_Q}. We employ analogous arguments and notation, starting with
\begin{align*}
\Lambda = \left\{\frac{e^{\lambda\cdot H + \log \rho}}{\tr(e^{\lambda\cdot H + \log \rho})}: \lambda \in \mathbb{R}^d \right\} \quad \text{and} \quad
M = \{\sigma: \tr(\sigma H) = m\}.
\end{align*}

\begin{lemma}\label{Lemma:minimizing_sigma}
Let $\rho \in \mathcal{D}(\mathcal{H})$ and $H_i, i \in [d]$ be observables on $\mathcal{H}$. Fix $m \in \mathbb{R}^d$. Then for any density operator $\sigma \in \mathcal{D}(\mathcal{H})$ satisfying $\tr(\sigma H)=m$, we have
\be
S(\sigma\|\rho) \geq \sup_{\lambda\in \mathbb{R}^d} \left[\lambda \cdot m - \log \tr(e^{\lambda \cdot H + \log \rho})\right].
\ee
Moreover the inequality is saturated if $\sigma = \sigma_{\lambda'} := e^{\lambda' \cdot H + \log \rho}/\tr(e^{\lambda' \cdot H + \log \rho)} \in \Lambda \cap M$ for some $\lambda' \in \mathbb{R}^d$:
\be
S(\sigma_{\lambda'}\|\rho) = \lambda' \cdot m - \log \tr(e^{\lambda' \cdot H + \log \rho}) = \sup_{\lambda \in \mathbb{R}^d} \left[\lambda \cdot m - \log \tr(e^{\lambda \cdot H + \log \rho}) \right].
\ee
\end{lemma}
\begin{proof}
Each $\lambda \in \mathbb{R}^d$ gives rise to a corresponding $\sigma_\lambda \in \Lambda$ (note that $\sigma_\lambda$ need not be in $M$). Then for any $\sigma$ satisfying $\tr(\sigma H)=m$, we have
\be
S(\sigma\|\rho)
&=& S(\sigma\|\sigma_{\lambda}) + \tr\{\sigma(\log \sigma_{\lambda}-\log \rho)\}\\ \nonumber
(\text{nonnegativity of $S(\sigma\|\rho)$}) &\geq& \tr\{\sigma (\log  (e^{\lambda \cdot H + \log \rho}) -\log \tr(e^{\lambda \cdot H + \log \rho}) -\log \rho)\}\\ \nonumber
&=& \tr\{\sigma (\lambda \cdot H)\} - \log \tr(e^{\lambda \cdot H + \log \rho})\\ \nonumber
&=& \lambda \cdot m- \log \tr(e^{\lambda \cdot H + \log \rho}).
\ee
Since this holds for all $\lambda \in \mathbb{R}^d$, we conclude that $S(\sigma\|\rho) \geq \sup_{\lambda \in \mathbb{R}^d}\left[\lambda \cdot m- \log \tr(e^{\lambda \cdot H + \log \rho})\right]$. Furthermore, if $\lambda' \in \mathbb{R}^d$ is such that $\sigma_{\lambda'} \in \Lambda \cap M$, then letting $\sigma=\sigma_{\lambda'}$ and rerunning the same argument sequence above gives
\begin{align*}
S(\sigma_{\lambda'}\|\rho) &= \tr\{\sigma_{\lambda'}(\log \sigma_{\lambda'}-\log \rho)\}\\
&= \tr\{\sigma_{\lambda'} (\log  (e^{\lambda' \cdot H + \log \rho}) -\log \tr(e^{\lambda' \cdot H + \log \rho}) -\log \rho)\}\\
&= \tr\{\sigma_{\lambda'} (\lambda' \cdot H)\} - \log \tr(e^{\lambda' \cdot H + \log \rho})\\
&= \lambda' \cdot m - \log \tr(e^{\lambda' \cdot H + \log \rho}).
\end{align*}
In particular, this also shows that $\lambda' = \argmax_{\lambda \in \mathbb{R}^d} \left[\lambda \cdot m - \log \tr(e^{\lambda \cdot H + \log \rho}) \right]$.
\end{proof}

\section{Block-encodings and Quantum Singular Value Transformation}\label{Appendix:BE+QSVT}
The technique of quantum signal processing \cite{low2016methodology} and its lifting via `qubitization' to quantum singular value transformation (QSVT) \cite{low2019hamiltonian,gilyen2019quantum} provide a concise way to formulate quantum algorithms, particularly for linear algebraic tasks. This framework has provided more efficient implementations of several existing quantum algorithms, such as Hamiltonian simulation \cite{low2017optimal,low2019hamiltonian}, amplitude amplification and estimation \cite{gilyen2019quantum,rall2023amplitude} and quantum linear systems solving \cite{gilyen2019quantum}, and even led to the discovery of new algorithms. For our purposes, we do not actually need the full generality of QSVT. As our matrices of interest are Hermitian and thus admit spectral decompositions, a relaxed version of QSVT---quantum \textit{eigenvalue} transformation (QET)---suffices. We direct readers interested in learning more about QSVT to \cite{gilyen2019quantum,martyn2021grand,dalzell2023quantum}.

\begin{definition}[Block-Encoding]\label{Definition:Block_encoding}
Let $A$ be an $n$-qubit matrix, $\alpha, \varepsilon \in \mathbb{R}_+$ and $a \in \mathbb{N}$. We say that the $(n+a)$-qubit unitary $U$ is an $(\alpha,a,\varepsilon)$-block-encoding of $A$ if
\[
\|A-\alpha(\bra{0^a}\otimes I_n)U(\ket{0^a}\otimes I_n)\| \leq \varepsilon.
\]
\end{definition}

\begin{remark}\label{Remark:Block_encoding}
Note that if $U$ is an $(\alpha,a,\varepsilon)$-BE of $A$, then equivalently it is a $(1,a,\frac{\varepsilon}{\alpha})$-BE of $\frac{A}{\alpha}$. Also, if we have a $(\alpha,a,\varepsilon)$-BE of $A$ then we also have a $(\alpha,a+a',\varepsilon+\varepsilon')$-BE of $A$, where $1 \leq a' \in \mathbb{N}$ and $\varepsilon' > 0$. Making the increment $a'$ simply corresponds to tacking on an extra $a'$-qubit identity operator $I_{a'}$. More specifically, if $U$ is an $(\alpha,a,\varepsilon)$-BE of $A$ then $I_{a'} \otimes U$ is an $(\alpha,a+a',\varepsilon)$-BE of $A$, since
\begin{align*}
\|A-\alpha(\bra{0^a}\otimes I_n)U(\ket{0^a}\otimes I_n)\| \leq \varepsilon \implies \|A-\alpha(\bra{0^{a'+a}}\otimes I_n)I_{a'} \otimes U(\ket{0^{a'+a}}\otimes I_n)\| \leq \varepsilon. 
\end{align*}
Finally, if $\varepsilon$ is already an error bound, $\varepsilon+\varepsilon'$ clearly serves as another error bound, albeit a weaker one.
\end{remark}

\cite{gilyen2019quantum} provides a construction of \textit{exact} block-encodings for density operators, assuming access to oracles which prepare the purifications of the density operators:
\begin{definition}[Purified quantum query-access]\label{Definition:Purified_access}
    Let $\rho$ be an $n$-qubit density operator. We say $\rho$ has purified quantum query-access if we have access to a $(n_\rho+n)$-qubit unitary operator $O_\rho$, where
    \begin{align*}
        O_\rho\ket{0^{n_\rho}}\ket{0^n} = \ket{\rho}
    \end{align*}
    prepares $\ket{\rho}$, the purification of $\rho$ (i.e. $\text{tr}_{n_\rho} \ket{\rho}\bra{\rho} = \rho$) with the help of $n_\rho$ ancilla qubits.\footnote{Theoretically, any $n$-qubit quantum state can be purified with at most $n$ ancilla qubits, so one can assume $n_\rho \leq n$. In practice however, it could be more convenient to use more than $n$ ancillas for purification. Thus we make the more relaxed assumption that $n_\rho = \poly(n)$.}
\end{definition}

\begin{proposition}[Block-encoding of density operators -- Lemma 45, \cite{gilyen2019quantum}]\label{Proposition:block_encoding_density_operators}
Let $\rho$ be an $n$-qubit density operator with purified quantum query-access via $O_\rho$. Then $\widetilde{O_\rho}:=(O_\rho^\dag \otimes I_n)(I_{n_\rho+n} \otimes \text{SWAP}_n)(O_\rho \otimes I_n)$ is a $(1,n+n_\rho,0)$-BE of $\rho$.
\end{proposition}

For general matrices which need not be density operators, \cite{chakraborty2018power,gilyen2019quantum} also showed how to implement their block-encodings efficiently, assuming the existence of quantum random access memory (QRAM) \cite{giovannetti2008quantum}. 
Given block-encodings of operators $A_i$, we can construct block-encodings of their linear combinations and products. For linear combinations, we make use of an auxiliary tool known as a `state preparation pair'. Recall that $\|\cdot\|_1$ is the $l_1$/Manhattan norm.
\begin{definition}[State Preparation Pair]\label{Definition:State_preparation_pair}
Let $y \in \mathbb{C}^m$ and $\|y\|_1 \leq \beta$. The pair of unitaries $(P_L,P_R)$ is called a ($\beta,b,\varepsilon_{\text{SP}}$)-state-preparation-pair for $y$ if
\begin{align*}
P_L\ket{0^b} = \sum_{j=0}^{2^b-1} c_j\ket{j}, \quad
P_R\ket{0^b} = \sum_{j=0}^{2^b-1} d_j\ket{j}
\end{align*}
such that $\sum_{j=0}^{m-1} |y_j-\beta c_j^*d_j| \leq \varepsilon_{\text{SP}}$ and $c_j^*d_j = 0$ for $j=m,\dots,2^b-1$.
\end{definition}
One can think of a state preparation pair as encoding the desired state/vector $y$ in the first $m$ elements of a length-$2^b$ column vector whose elements are $c_j^*d_j$, up to an error of $\varepsilon_{\text{SP}}$. The role of $\beta$ is to take care of normalization.

\begin{proposition}[Linear combination of block-encoded matrices -- Lemma 52, \cite{gilyen2019quantum}]\label{Proposition:GSLW_linear_combination_BE}
Let 
\begin{enumerate}[i.]
\item $A_j,\; j=0,\dots,m-1$ be $n$-qubit operators with respective ($\alpha,a,\varepsilon_{\text{BE}}$)-BEs $U_j$,
\item $A = \sum_{j=0}^{m-1} y_jA_j$ for $y := (y_0,\dots,y_{m-1}) \in \mathbb{C}^m$,
\item $(P_L,P_R)$ be a $(\beta,b,\varepsilon_{\text{SP}})$-state-preparation-pair for $y$.
\end{enumerate}
Then there exists a $(\alpha\beta,a+b,\alpha\varepsilon_{\text{SP}}+\beta\varepsilon_{\text{BE}})$-BE of $A$, given by
\[
\widetilde{W} = (P_L^\dag \otimes I_a \otimes I_n)W(P_R \otimes I_a \otimes I_n),
\]
where
\[
W = \sum_{j=0}^{m-1} \ket{j}\bra{j} \otimes U_j + \sum_{j=m}^{2^b-1} \ket{j}\bra{j} \otimes I_a \otimes I_n
\]
is a $(n+a+b)$-qubit unitary.
\end{proposition}

In Proposition \ref{Proposition:GSLW_linear_combination_BE}, the subnormalization factors of the $A_j$'s are restricted to be the same. Later on, we will need a slight generalization of the above result whereby this requirement is dropped.

\begin{proposition}[Generalized linear combination of block-encoded matrices]\label{Proposition:General_linear_combination_BE}
Let 
\begin{enumerate}[i.]
\item $A_j,\; j=0,\dots,m-1$ be $n$-qubit operators with respective ($\alpha_j,a,\varepsilon_{\text{BE}}$)-BEs $U_j$ for $\alpha := (\alpha_0,\dots,\alpha_{m-1}) \in \mathbb{C}^m$,
\item $A = \sum_{j=0}^{m-1} y_jA_j$ for $y := (y_0,\dots,y_{m-1}) \in \mathbb{C}^m$,
\item $(P_L,P_R)$ be a $(\beta,b,\varepsilon_{\text{SP}})$-state-preparation-pair for $\alpha \odot y$.
\end{enumerate}
Then there exists a $(\beta,a+b,\frac{\beta}{\inf_j \alpha_j}\varepsilon_{\text{BE}}+\varepsilon_{\text{SP}})$-BE of $A$, given by
\[
\widetilde{W} = (P_L^\dag \otimes I_a \otimes I_n)W(P_R \otimes I_a \otimes I_n),
\]
where
\[
W = \sum_{j=0}^{m-1} \ket{j}\bra{j} \otimes U_j + \sum_{j=m}^{2^b-1} \ket{j}\bra{j} \otimes I_a \otimes I_n
\]
is a $(n+a+b)$-qubit unitary.
\end{proposition}
\begin{proof}
The following is adapted from the proof of Lemma 52, \cite{gilyen2019quantum}. By definition of state-preparation pairs (see Definition \ref{Definition:State_preparation_pair}), $P_L\ket{0^b} = \sum_{j=0}^{2^b-1} c_j\ket{j}$ and $P_R\ket{0^b} = \sum_{j=0}^{2^b-1} d_j\ket{j}$ such that $\sum_{j=0}^{m-1}|\alpha_j y_j - \beta c_j^*d_j| \leq \varepsilon_{\text{SP}}$. First we evaluate the block extraction of $\widetilde{W}$. We have

\onecolumngrid
\begin{align*}
&\quad\;(\bra{0^{b+a}}\otimes I_n)\widetilde{W}(\ket{0^{b+a}}\otimes I_n)\\
&= (\bra{0^{b+a}}\otimes I_n)(P_L^\dag \otimes I_a \otimes I_n)\left( \sum_{j=0}^{m-1} \ket{j}\bra{j} \otimes U_j + \sum_{j=m}^{2^b-1} \ket{j}\bra{j} \otimes I_a \otimes I_n \right)(P_R \otimes I_a \otimes I_n)(\ket{0^{b+a}}\otimes I_n)\\
&= \sum_{j=0}^{m-1} \bra{0^b}P_L^\dag \ket{j}\bra{j} P_R\ket{0^b} \cdot (\bra{0^a}\otimes I_n)U_j(\ket{0^a}\otimes I_n)\\
&= \sum_{j=0}^{m-1} c_j^*d_j \cdot (\bra{0^a}\otimes I_n)U_j(\ket{0^a}\otimes I_n).
\end{align*}

In going from the first equality to the second, we have made use of the fact that for state preparation pairs $c_j^*d_j=0$ for $j=m,\dots,2^b-1$. The second summand in $W$ is thus annihilated.
Therefore,
\begin{eqnarray}
\left\| A - \beta (\bra{0^{b+a}}\otimes I_n)\widetilde{W}(\ket{0^{b+a}}\otimes I_n) \right\| &=& \left\| A - \sum_{j=0}^{m-1} (\beta c_j^*d_j - \alpha_jy_j + \alpha_jy_j) 
\cdot (\bra{0^a}\otimes I_n)U_j(\ket{0^a}\otimes I_n) \right\| \\\nonumber
&\leq& \sum_{j=0}^{m-1} |\beta c_j^*d_j - \alpha_jy_j| + \left\| A - \sum_{j=0}^{m-1} \alpha_jy_j(\bra{0^a}\otimes I_n)U_j(\ket{0^a}\otimes I_n) \right\|\\\nonumber
&\leq& \varepsilon_{\text{SP}} + \left\| \sum_{j=0}^{m-1} y_jA_j - \sum_{j=0}^{m-1} y_j\alpha_j(\bra{0^a}\otimes I_n)U_j(\ket{0^a}\otimes I_n) \right\| \\\nonumber
&\leq& \varepsilon_{\text{SP}} + \sum_{j=0}^{m-1} |y_j| \left\| A_j - \alpha_j (\bra{0^a}\otimes I_n)U_j(\ket{0^a}\otimes I_n) \right\|\\\nonumber
&\leq& \varepsilon_{\text{SP}} + \sum_{j=0}^{m-1} |y_j| \varepsilon_{\text{BE}} \leq \varepsilon_{\text{SP}} + \frac{\beta}{\inf_j \alpha_j}\varepsilon_{\text{BE}}.
\end{eqnarray}
where the last inequality was obtained using $\beta \geq \sum_{j=0}^{m-1}|\alpha_j y_j| \geq \sum_{j=0}^{m-1}(\inf_k \alpha_k)|y_j|$.
\end{proof}

\begin{remark}
In the special case where the block-encodings of the $A_j$'s have the same subnormalization factors, i.e., $\alpha_j=\alpha$ for all $j$, we recover Proposition \ref{Proposition:GSLW_linear_combination_BE} from Proposition \ref{Proposition:General_linear_combination_BE} . To see this, observe that if $(P_L,P_R)$ is a $(\beta,b,\varepsilon_{\text{SP}})$-state-preparation-pair for $\alpha \odot y$, then $\sum_j |\alpha_jy_j - \beta c_j^*d_j| \leq \varepsilon_{\text{SP}} \implies \sum_j |\alpha y_j - \beta c_j^*d_j| \leq \varepsilon_{\text{SP}} \implies \sum_j |y_j - \frac{\beta}{\alpha} c_j^*d_j| \leq \frac{\varepsilon_{\text{SP}}}{\alpha}$, thus implying $(P_L,P_R)$ is a $(\frac{\beta}{\alpha},b,\frac{\varepsilon_{\text{SP}}}{\alpha})$-state-preparation-pair for $y$. According to Proposition \ref{Proposition:GSLW_linear_combination_BE}, $\widetilde{W}$ is then a $(\alpha \cdot \frac{\beta}{\alpha},\;a+b,\;\alpha \cdot \frac{\varepsilon_{\text{SP}}}{\alpha} + \frac{\beta}{\alpha}\varepsilon_{\text{BE}})$-BE of $A$. This is in agreement with Proposition \ref{Proposition:General_linear_combination_BE}.
\end{remark}

We now arrive at a milestone within the QSVT framework. Namely, the ability to implement block-encodings of polynomials of a matrix from a given block-encoding of the matrix. In many applications however, the functions of interest are not polynomials. In such cases, one has to first approximate the desired function by a polynomial in order to apply QSVT/QET.

\begin{theorem}[Polynomial Eigenvalue Transformation -- Theorem 56, \cite{gilyen2019quantum}]\label{Theorem:qet_of_hermitian_matrices}
Let $U$ be an $(\alpha,a,\varepsilon)$-encoding of a Hermitian matrix $A$ (equivalently, a $(1,a,\varepsilon/\alpha)$-encoding of $A/\alpha$) and $P \in \mathbb{R}[x]$ be a degree-$d$ polynomial satisfying $|P(x)| \leq \frac{1}{2}$ on $[-1,1]$.
Then, one can construct a quantum circuit $\tilde{U}$ which is a $(1,a+2,4d\sqrt{\varepsilon/\alpha})$-encoding of $P(A/\alpha)$. $\tilde{U}$ consists of $d$ $U$ and $U^\dag$ gates, one controlled-$U$, and $\mathcal{O}((a+1)d)$ other one- and two-qubit gates.
\end{theorem}

\begin{proposition}[Bounded Polynomial Approximation -- Corollary 66, \cite{gilyen2019quantum}]\label{Proposition:gslw_polynomialapproximation}
Let $x_0 \in [-1,1]$, $r \in (0,2]$, $\delta \in (0,r]$ and let $f: [x_0-r-\delta, x_0+r+\delta] \longrightarrow \mathbb{C}$ be such that $f(x) = \sum_{l=0}^\infty a_l(x-x_0)^l$ for all $x \in [x_0-r-\delta,x_0+r+\delta]$. Suppose $B>0$ is such that $\sum_{l=0}^\infty (r+\delta)^l|a_l| \leq B$. Let $\varepsilon \in (0,\frac{1}{2B}]$, then there is an efficiently computable polynomial $P \in \mathbb{C}[x]$ of degree $\mathcal{O}\left( \frac{1}{\delta} \log \left( \frac{B}{\varepsilon} \right) \right)$ such that
\begin{align}
&\|f(x)-P(x)\|_{[x_0-r,x_0+r]} \leq \varepsilon\\
&\|P(x)\|_{[-1,1]} \leq \varepsilon+\|f(x)\|_{[x_0-r-\delta/2,x_0+r+\delta/2]} \leq \varepsilon+B\\
&\|P(x)\|_{[-1,1]\setminus[x_0-r-\delta/2,x_0+r+\delta/2]} \leq \varepsilon.
\end{align}
If we choose $B$ sufficiently large such that $\frac{1}{2B} < 1$, then we also have an $\varepsilon$-independent bound on $P(x)$: $\|P(x)\|_{[-1,1]} \leq 1+B$.
\end{proposition}

Theorem \ref{Theorem:qet_of_hermitian_matrices} and Proposition \ref{Proposition:gslw_polynomialapproximation} are to be used in conjunction to produce block-encodings of general functions of Hermitian matrices. In doing so, we first note that Theorem \ref{Theorem:qet_of_hermitian_matrices} produces an encoding of $P(A/\alpha)$, not $P(A)$. Thus, with a polynomial approximation of $f$, say $P(x) \approx f(x)$, it is generally not true that $P(A/\alpha) \approx f(A)$. What we need is a polynomial approximation not of $f$, but of a (horizontally) scaled version of $f$, 
$f'(x) := f(\alpha x)$, so that $P(x) \approx f'(x) \implies P(A/\alpha) \approx f'(A/\alpha) = f(A)$. Second, we also have to take into account the polynomial approximation error incurred in producing the final desired block encoding $f(A)$. 
We take care of these matters in Corollary \ref{Corollary:BE_of_hermitian_matrices}, which, given the block-encoding of an arbitrary Hermitian matrix $A$, produces a block-encoding of $f(A)$, where $f$ is a generic real-valued function.

\begin{corollary}[Block-encoding functions of general Hermitian matrices]\label{Corollary:BE_of_hermitian_matrices}
Given
\begin{enumerate}[i.]
\item A Hermitian matrix $\lambda_{\min} \leq A \leq \lambda_{\max}$, $-\infty < \lambda_{\min} < \lambda_{\max} < \infty$ and $U$, an $(\alpha,a,\varepsilon)$-encoding of $A$.
\item $f: I \longrightarrow \mathbb{R}$, a smooth function on an open interval $I$ containing  $[\lambda_{\min},\lambda_{\max}]$. Assume the function $x \mapsto f(\alpha x)$ satisfies the conditions in Proposition \ref{Proposition:gslw_polynomialapproximation} with $[\lambda_{\min}/\alpha,\lambda_{\max}/\alpha] \subseteq [x_0-r,x_0+r]$ and series-of-coefficients bound $B$.
\item Polynomial approximation error tolerance for $f$: $\varepsilon_\text{poly} \in (0,\frac{1}{2}]$.
\end{enumerate}
Then there exists a quantum circuit $U_f$ which is a $\left( 2(1+B),\; a+2,\; \varepsilon_\text{poly} + 2(1+B)(4d\sqrt{\varepsilon/\alpha}) \right)$-encoding of $f(A)$. The construction of $U_f$ makes $d = \mathcal{O}\left( \frac{1}{\delta} \log \frac{B}{\varepsilon_\text{poly}} \right)$ queries to $U$.
\end{corollary}

\begin{proof}
First, $\alpha \geq \|A\| = \max \{|\lambda_{\min}|,|\lambda_{\max}|\}$. Define the scaling map $t_\alpha : x \mapsto x/\alpha$, so that under this map $[\lambda_{\min},\lambda_{\max}] \mapsto [\lambda_{\min}/\alpha,\lambda_{\max}/\alpha]$. By assumption on $f$ there exists $x_0 \in [-1,1]$, $r \in (0,2]$, $\delta \in (0,r]$ such that (i.) $[\lambda_{\min}/\alpha,\lambda_{\max}/\alpha] \subseteq [x_0-r,x_0+r]$, (ii.) $f \circ t_\alpha^{-1}(x) = \sum_{l=0}^\infty a_l(x-x_0)^l$ on $[x_0-r-\delta,x_0+r+\delta]$ and (iii.) $\sum_{l=0}^\infty (r+\delta)^l|a_l| \leq B$ for some $B>0$.

By Proposition \ref{Proposition:gslw_polynomialapproximation}, given polynomial approximation error tolerance $\varepsilon_\text{poly}$ there exists a polynomial $Q \in \mathbb{C}[x]$ of degree $\mathcal{O}\left( \frac{1}{\delta} \log \left( \frac{B}{\varepsilon_\text{poly}} \right) \right)$ which $\varepsilon_\text{poly}$-approximates $f \circ t_\alpha^{-1}$ on $[x_0-r,x_0+r]$ and is bounded above by $1+B$ on $[-1,1]$. Since $\|A/\alpha\| \in [\lambda_{\min}/\alpha,\lambda_{\max}/\alpha] \subseteq [x_0-r,x_0+r] $, we have
\begin{align*}
\left\|f \circ t_\alpha^{-1}\left(\frac{A}{\alpha}\right) - Q\left(\frac{A}{\alpha}\right)\right\| \leq \|f \circ t_\alpha^{-1}(x) - Q(x)\|_{[x_0-r,x_0+r]} \leq \varepsilon_\text{poly}.
\end{align*}
In order to apply Theorem \ref{Theorem:qet_of_hermitian_matrices}, our polynomial has to be real and upper-bounded by $1/2$ on $[-1,1]$. Observe that for any complex-valued function $F$ and domain $S$, 
\begin{align*}
\|F\|_S = \sup_{x \in S} |F(x)| = \sup_{x \in S} \sqrt{(\re F(x))^2 + (\im F(x))^2} \geq \sup_{x \in S} |\re F(x)| = \|\re F\|_S.
\end{align*}
Since $f$ itself is real-valued, $\re Q \in \mathbb{R}[x]$ is qualified to assume the role of $P$ in Proposition \ref{Proposition:gslw_polynomialapproximation}. That is, the real polynomial $\re Q$ also $\varepsilon_\text{poly}$-approximates $f \circ t_\alpha^{-1}$ on $[x_0-r,x_0+r]$ and is bounded above by $1+B$ on $[-1,1]$. Thus, letting $P \leftarrow \frac{\re Q}{2(1+B)}$ in Theorem \ref{Theorem:qet_of_hermitian_matrices} we obtain $\tilde{U}$, a $(1,a+2,4d\sqrt{\varepsilon/\alpha})$-encoding of $\frac{\re Q}{2(1+B)}(A/\alpha)$, where $d = \mathcal{O}\left( \frac{1}{\delta} \log \left( \frac{B}{\varepsilon_\text{poly}} \right) \right)$. Putting these together and noting that $f \circ t_\alpha^{-1}(\frac{A}{\alpha}) = f(A)$, we have
\begin{align*}
\left\|\frac{f(A)}{2(1+B)} - (\bra{0^{a+2}}\otimes I)\tilde{U}(\ket{0^{a+2}}\otimes I)\right\| &\leq \left\| \frac{f \circ t_\alpha^{-1}(\frac{A}{\alpha})}{2(1+B)}-\frac{\re Q(\frac{A}{\alpha})}{2(1+B)} \right\| +\left\| \frac{\re Q(\frac{A}{\alpha})}{2(1+B)} - (\bra{0^{a+2}}\otimes I)\tilde{U}(\ket{0^{a+2}}\otimes I) \right\|\\
&\leq \frac{\varepsilon_\text{poly}}{2(1+B)} + 4d\sqrt{\varepsilon/\alpha}.
\end{align*}
Thus, choosing $U_f = \tilde{U}$ gives us a $\left( 2(1+B),\; a+2,\; \varepsilon_\text{poly} + 2(1+B)(4d\sqrt{\varepsilon/\alpha}) \right)$-encoding of $f(A)$.
\end{proof}

\end{appendix}

\end{document}